\documentclass[mnsc,nonblindrev]{informs3}

\OneAndAHalfSpacedXI

\usepackage[plainpages=false,hyperfootnotes=false]{hyperref}
\hypersetup{
  colorlinks   = true, 
  urlcolor     = blue, 
  linkcolor    = red, 
  citecolor   = blue 
}

\usepackage{natbib}
 \bibpunct[, ]{(}{)}{,}{a}{}{,}%
 \def\bibfont{\small}%
\TheoremsNumberedThrough     
\ECRepeatTheorems

\EquationsNumberedThrough    

\MANUSCRIPTNO{}

\usepackage{dsfont}
\usepackage{amssymb}
\usepackage{amsmath}
\usepackage{enumitem}
\usepackage{booktabs}
\usepackage{siunitx}
\usepackage[colorinlistoftodos]{todonotes}
\usepackage{cleveref}
\usepackage{graphicx}

\newcommand{\X}{{\boldsymbol{X}}}
\newcommand{\D}{{\boldsymbol{D}}}

\renewcommand{\d}{{\boldsymbol{d}}}
\renewcommand{\t}{{\boldsymbol{t}}}

\newcommand{\abs}[1]{\left\lvert #1 \right\rvert}  

\newcommand{\A}{\mathsf{A}}
\newcommand{\M}{{\mathcal{M}}}
\newcommand{\mN}{{\mathcal{N}}}

\newcommand{\F}{\mathfrak{F}}

\newcommand{\mfu}{{\mathfrak{u}}}

\newcommand{\Cov}{\mathrm{Cov}}

\newcommand{\bt}{{\boldsymbol{t}}}
\newcommand{\ep}{{\varepsilon}}
\newcommand{\Id}{{\mathds{1}}}

\newcommand{\E}{{\mathbb{E}}}
\newcommand{\Ex}{\mathbb{E}_{\boldsymbol{x}}}

\renewcommand{\P}{{\mathbb{P}}}
\newcommand{\Q}{{\mathbb{Q}}}
\newcommand{\Qb}{{\mathbb{Q}^{\mathfrak{b}}}}
\newcommand{\Qd}{{\mathbb{Q}^{\dagger}}}
\newcommand{\Qtilde}{{\Tilde{\mathbb{Q}}}}
\newcommand{\R}{{\mathds{R}}}
\newcommand{\Lp}{{\mathcal{L}^2}}

\newcommand{\e}{{\boldsymbol{e}}}

\newcommand{\bK}{{\boldsymbol{K}}}
\newcommand{\x}{{\boldsymbol{x}}}

\newcommand{\W}{{\boldsymbol{W}}}

\newcommand{\Z}{{\boldsymbol{Z}}}

\newcommand{\bfeta}{{\boldsymbol{\eta}}}

\renewcommand{\Finv}{{\breve{F}}}

\newcommand{\supp}[1]{{\mathrm{supp}\{#1\}}}

\usepackage[plainpages=false,hyperfootnotes=false]{hyperref}
\hypersetup{
  colorlinks   = true, 
  urlcolor     = blue, 
  linkcolor    = red, 
  citecolor   = blue 
}

\begin{document}

\RUNAUTHOR{Miao and Pesenti}

\RUNTITLE{Discrimination-Insensitive Pricing}

\TITLE{Discrimination-Insensitive Pricing}

\ARTICLEAUTHORS{%
\AUTHOR{Kathleen E.~Miao}
\AFF{Department of Statistical Sciences, University of Toronto, Canada, \EMAIL{k.miao@mail.utoronto.ca}} 
\AUTHOR{Silvana M.~Pesenti}
\AFF{Department of Statistical Sciences, University of Toronto, Canada, \EMAIL{silvana.pesenti@utoronto.ca}} 
\vspace{0.5em}
March 27, 2026
} 

\ABSTRACT{Rendering fair prices for financial, credit, and insurance products is of ethical and regulatory interest. In many jurisdictions, discriminatory covariates, such as gender and ethnicity, are prohibited from use in pricing such instruments.  In this work, we propose a discrimination-insensitive pricing framework, where we require the pricing principle to be insensitive to the (exogenously determined) protected covariates, that is the sensitivity of the pricing principle to the protected covariate is zero. We formulate and solve the optimisation problem that finds the nearest (in Kullback-Leibler (KL) divergence) ``pricing'' measure to the real world probability, such that under this pricing measure the principle is discrimination-insensitive.  We call the solution the discrimination-insensitive measure and provide conditions for its existence and uniqueness.
In situations when there are more than one protected covariates, the discrimination-insensitive pricing measure might not exist, and we propose a two-step procedure. First, for each protected covariate separately, we find the measure under which the pricing principle becomes insensitivity to that covariate. Second we reconcile these measures through a constrained barycentre model. We provide a close-form solution to this problem and give conditions for existence and uniqueness of the constrained barycentre pricing measure. As an intermediary result, we prove the representation, existence, and uniqueness of the KL barycentre of general probability measures, which may be of independent interest.
Finally, in a numerical illustration, we compare our discrimination-insensitive premia and the constrained barycentre pricing measure with recently proposed fair premia from the actuarial literature.
}%


\KEYWORDS{fairness, discrimination, sensitivity, barycentre, distortion risk measures, Kullback-Leibler} 

\maketitle

%



\section{Introduction}
\vspace{1em}
\subsection{Overview and contribution}
Ensuring a fair price is of critical interest in many domains of finance, insurance, and risk management.  In general, ``fairness" connotes that prices are statistically accurate (determined by risk), socially equitable (not determined by inherent characteristics of the client), and transparent (tractable and interpretable). Translating these notions of fairness into a mathematical setting is a lively discussion both in academia and practice.  In this work, we take ``fair", or ``free from discrimination" to be mathematically defined as the pricing principle having zero sensitivity with respect to protected characteristics.  The zero sensitivity means that the pricing principle is not impacted by distributional perturbations of the protected covariate; resulting in an insensitive (to the protected covariate) pricing principle. We consider the setting where a financial firm uses a distortion risk principle to make a risk-informed decision upon a loss model, under the exogenously imposed restriction that some characteristics are prohibited from use in pricing. Typically, pricing is a two stage framework; first, information collected from the client is used to find a reference model $\P$ that represents their losses, perhaps using a generalised linear model or more sophisticated machine learning techniques, and second pricing is performed on this model. Thus, concerns arise in the use of protected characteristics in both the modelling and the pricing portion of this procedure. In this work, we allow for all covariates to be used in estimating the loss model, however the resulting pricing framework should be insensitive to the protected covariates. Allowing the protected covariates in the modelling step results into more accurate forecasting and understanding of the risks, while requiring insensitivity of the pricing principle results in similar prices for customers with similar permitted covariates. More specifically, our notion of \textit{discrimination-insensitivity} is inspired by sensitivity analysis, as we interpret the need to be ``free of discrimination" as requiring \textit{insensitivity to the discriminatory covariate}. That is, discrimination in this work is characterised by a differential sensitivity of the risk pricing principle in direction of protected covariates, which quantifies the impact of distributional perturbations on the pricing principle.

As pricing principles are in general sensitive to protected covariates, we consider alternative probability measures such that the sensitivity of the pricing principle under these alternative probability measures is rendered null. Additionally to the insensitivity, we further impose a conditional expectation constraint, which guarantees that the conditional expected losses under the new pricing measures equal those under the reference model $\P$.  This is done to avoid cross-subsidies, and to ensure that losses are, on average, covered. Thus, our proposed discrimination-free pricing framework is suitable for prices calculated under the real world probability, e.g., for credit and loan modelling and insurance premia calculations.  

We formulate a distributional optimisation problem where we find the probability measures under which zero sensitivity of the pricing principle with respect to each of the protected covariates is  enforced,  the conditional expectation of the losses is  maintained, and which is as near as possible, in Kullback-Leibler (KL) divergence, to the reference measure $\P$. The KL divergence is chosen as it maintains the support of the reference measure, a desirable property as many protected characteristics are categorical, for example, gender and ethnicity.  We solve this distributional optimisation problem, whose solution we term the \textit{discrimination-insensitive pricing measure}, provide a representation of its change-of-measure and prove, under mild conditions, its existence and uniqueness.

A second key contribution of our work is how we treat the scenarios in which the discrimination-insensitive pricing measure is too restrictive, infeasible, or otherwise undesirable.  This may arise if too many characteristics are protected, if the discrimination-insensitive pricing measure is too far from the reference measure, or if it does not exist. To do this, we modify the original optimisation problem into a two-step pricing procedure. First, for each discriminatory covariate separately, we derive a \textit{marginally-insensitive measure}, which is the nearest measure in KL divergence to $\P$ that is insensitive to that protected covariate. Then, in the second step,  the marginally-insensitive measures are merged via a constrained barycentre method, where in the second step we impose the conditional expectation constraint. We prove a succinct representation of the constrained barycentre pricing measure and prove existence and uniqueness. We also solve the pure KL barycentre problem in a general probability measure setting, i.e., without the expectation constraint, which may be of independent interest. 

\subsection{Relation to existing literature}
\label{subsec: lit}

Our methodology is distinct from the literature in that we consider simultaneously the impact of multiple discriminatory covariates, the insensitivity together with the conditional expectation constraint, and that the pricing principles belong to the family of distortion risk principles rather of the pure premium typically studied in the literature. Moreover, the use of barycentre models to overcome difficulties with multiple protected covariates is novel. Furthermore, our pricing measures are tractable, unique, and interpretable.

Fairness is a growing topic of study, with advances in the contexts of machine learning and algorithmic fairness \citep{OnetoChiappa2020, KallusMaoZhou2022},  determining creditworthiness, such as for credit scoring and making lending decisions, \citep{AndreevaAnsellCrook2004, AndreevaMatuszyk2019, Hurlin2026}, in pension design and pricing \citep{AlonsoBoadoDevolder2018}, access to financial services \citep{FuLiu2023, ScottEtAl2024}, and in pricing insurance and underwriting.  Many approaches such as equalized-odds, avoiding proxy discrimination, and excluding the protected characteristics have been proposed.  In the insurance context, the fairness of a policy is pertinent to actuarial researchers and practitioners alike, as governments and governing bodies have implemented regulations and laws to enforce various fairness criterion, e.g. prohibiting discrimination from the results of genetic testing in Canada \citep{Canada_Genetic_NonDiscrimination_Act_2017}, gender in the European Union \citep{EU_Directive_2004_113_EC}, and we refer to  \cite{XinHuang2024} for a thorough treatment of the current regulatory landscape. Further calls have been made by various actuarial societies for methods to mitigate and analyse the discrimination in current practices \citep{CanadianInstituteofActuaries2023, SocietyofActuaries2022, CasualtyActuarialSociety2025}.  Many responses have been made in the literature in the context of fair insurance pricing.  A seminal work, \cite{LindholmEtAl2022discrimfree}, was the first to  mathematically define direct and indirect discrimination, a main notion of discrimination in the literature. Recent advances include \cite{CoteCoteCharpentier2024}, who address indirect discrimination from a casual perspective, \cite{Machado_Charpentier_Gallic_2025} consider sequential transport to assess counterfactual fairness, and \cite{LindholmEtAl2024sensitivity}, where authors use variance-based sensitivities to measure proxy discrimination. \cite{FreesHuang2023} consider the appropriateness of actuarial discrimination from an economic and social justice lens, and \cite{Charpentier2024} provides a comprehensive overview of discrimination and bias in insurance, for an actuarial audience.  This body of work is distinct from works relating to discrimination in machine learning communities, where a typical task is classification as opposed to making a decision based on risk.  With many references available in this line of literature, we refer the reader to a survey \cite{MehrabiEtAl2021} and references therein.  Moreover, while the use of sensitivity analysis is well established in risk management and finance, see \cite{BorgonovoPlischke2016} for an overview, using sensitivity to quantify fairness is relatively new with the exception of \cite{HuangPesenti2025}. The authors \cite{HuangPesenti2025} consider the problem of modifying the pricing rule to mitigate sensitivity to protected covariates. Our approach is distinct in that we derive a discrimination-free pricing measures while not changing the pricing principal, have a conditional expectation constraint, and propose a barycentre approach for situations with multiple protected covariates; the latter is a novel approach in the fairness literature.

\subsection{Structure of the article}
\Cref{sec: preliminaries} defines the discrimination-insensitive pricing measure and the discrimination-insensitive pricing principle.  In \Cref{sec: disc-ins-pricing-measure}, we state our main results, prove the existence and uniqueness of the discrimination-insensitive pricing measure, and demonstrate its behaviour in an example.  \Cref{sec: barycentre} is devoted to the constrained barycentre, that is we first derive marginal insensitive measures and second reconcile them using a constrained barycentre approach. We further provide the representation of the KL barycentre of probability measures, and prove its existence and uniqueness, which may be of independent interest.  Finally, in \Cref{sec: application} we preform a case study to illustrate the behaviour of the risk principle under the pricing measures introduced in this work.

\section{Quantifying discrimination via sensitivity}
\label{sec: preliminaries}
\vspace{1em}

\subsection{Prices and their sensitivity}
Let $(\Omega, \F, \P)$ be a probability space where $\P$ is a probability measure, and let $\Lp := \Lp(\Omega, \F, \P)$ be the space of square integrable random variables (rvs) on $(\Omega, \F, \P)$. For $Y \in \Lp$, we denote by $\supp{Y}$ the support of $Y$ and write $Y \sim F_Y$, where $F_Y$ is the cumulative distribution function (cdf) of $Y$ under $\P$, i.e., $F_Y(y):= \P(Y\le y)$, and further let $\Finv_Y(q) := \inf \{y \in \R : q \leq F_Y(y) \}$, $q \in (0,1)$, be the (left-)quantile function of $Y$. The cumulant generating function (cgf) of a rv $Y$ is defined as $K_Y(t) := \log \E[\e^{tY}]$, while for a random vector $\Z$, we define its cgf as $\mathbf{K}_{\Z}(\t) = \log(\E[\e^{\t^\intercal \Z}])$. Throughout, equalities and inequalities of rvs are meant to be component-wise and in a $\P$-a.s. sense. Moreover, we interpret the probability measure $\P$ as the real-world probability measure, i.e. the $\P$-distribution of a loss $Y$ corresponds to the loss observed from data.

Personal characteristics of policyholders are collected for the purposes of calculating their premia. However, in some jurisdictions, certain characteristics might be protected from being used to estimate premia \citep{Billingsley2022, XinHuang2024}. These prohibitions are typically determined by local regulation, legislation, or best practices. We denote the protected characteristics of the insured by $\D := (D_1, \ldots, D_m)$, where $D_i \in \Lp$ for all $i \in \M := \{1, \ldots, m\}$, and call $\D$ the vector of discriminatory covariates. The other covariates, that is the permitted characteristics, $\X:= (X_1, \ldots, X_n)$, where $X_i \in \Lp$ for all $i \in \mN := \{1, \ldots, n\}$, are called the non-discriminatory covariates. We assume that the split into discriminatory and non-discriminatory covariates is exogenously given.

We denote the policy holder claim, or aggregate loss, by $Y := h(\X,\D) + \mfu$, where $h:\R^{n+m} \to \R$ is assumed to be differentiable and determined by the insurer, and $\mfu$ is a random error term.
The aggregation function $h$ can be estimated by practitioners when modelling losses using both permitted and discriminatory covariates. Examples include regression models or more sophisticated machine learning techniques. Throughout, we assume that $h$ is non-constant.

An insurer is provided with policyholder's information $(\x, \d) $ -- realisations of $(\X, \D)$ -- to calculate their premium, however, the insurer can only use the non-discriminatory covariates for pricing. Thus, the premium is calculated via the conditional rv $Y$ given $\X = \x$. Thus, for fixed $\x$, we denote by $Y|_\x$ the rv $Y$ given $\X=\x$. Here, we focus on pricing regimes that belong to the class of distortion risk principle, defined below.  Distortion risk principles are popular due to their attractive properties such as law-invariance, monotonicity, comonotone additivity, see e.g.,  \cite{WangYoungPanjer1997}. Distortion risk principles include the well-known Value-at-Risk (VaR), Expected Shortfall (ES), as well as the Wang transform, and the expected value. 

\begin{definition}[Distortion risk principle]
The distortion risk principle for covariates $\x$ with losses $Y|_\x \in \Lp$ is 
\[
\rho\big(Y|_\x\big):= \int_0^1 \Finv_{Y|_\x}(u) \;\gamma(u) du = \E[Y\; \gamma(U_{Y|_\x}) ~|~ \X = \x]\,,
\]
where $\gamma: [0,1] \to [0,+\infty)$ satisfies $\int_0^1 \gamma(u) du = 1$ and is called a distortion weight function, and where $U_{Y|\x}:= F_{Y|\x}(Y|_\x)$ is, conditional on $\X = \x$, a uniform random variable under $\P$.
\end{definition}
For ease of notation, we write $\E_\x[\;\cdot\;]:= \E[\;\cdot ~|~\X=\x]$ to denote the expectation conditional on the permitted covariate $\X = \x$. 

The literature on  \textit{discrimination-free} or \textit{fairness} in insurance is extensive and different notations of fairness and discrimination-free have been proposed. For example, \cite{FreesHuang2023} discuss that fairness depends on the pooling type of the insured group and the provider. In \cite{LindholmEtAl2022discrimfree}, the authors' interpretation of discrimination-free is that a premium should, in a probabilistic sense, avoid both direct and indirect discrimination.  In this work, we take the stance that \textit{discrimination-free} means \textit{insensitivity to the discriminatory covariate}. The sensitivity to the discriminatory covariate quantifies how much the premium changes with respect to perturbations to that discriminatory covariate. Mathematically, we define the \textit{sensitivity} of the distortion risk principle to a discriminatory covariate via the G{\^a}teaux derivative.

\begin{definition}[Sensitivity of the distortion risk principle]\label{def:sensitivity}
Let $\partial_{D_i} \rho(Y|_\x)$ denote the G{\^a}teaux derivative of the distortion risk principle $\rho(Y|_\x)$ with respect to the $i$-th discriminatory covariate $D_i$ in direction $D_i$, that is
\[
\partial_{D_i} \rho\big(Y|_\x\big):= \lim_{\ep \to 0} \frac{ \rho\big(h(\X, \D_{i,\ep}) + \mfu ~\big|~\X = \x\big) - \rho\big(h(\X, \D) + \mfu ~\big|~\X = \x\big)}{\ep}\;,
\]
where $\D_{i,\ep}:=(D_1, \ldots, D_{i-1}, D_{i,\ep}, D_{i+1}, \ldots, D_d)$, and $D_{i,\ep}$ is a rv comonotonic to $D_i$ with cdf $F_{D_{i,\ep}}(\cdot):= F_{D_i}\big(\frac{\cdot}{1+\ep}\big)$.
\end{definition}

If there is only one discriminatory covariate or the index is clear from context, we simply write $\partial \rho(Y|_\x)$. Note that the perturbation of $D_i$ corresponds to a proportional modification $D_{i,\ep} = D_i(1 +\ep)$. While alternative perturbations may be considered (see e.g., \cite{PesentiMillossovichTsanakas2025} for sensitivities to different perturbations), we restrict here to those of proportional type. As we consider the sensitivity of distortion risk principles to a perturbation of a protected covariate, the perturbation should modify the entire distribution of $D_i$ and be non-parametric, which is achieved by the proportional perturbation. The proportional perturbation $D_{i,\ep}$ is conceptually intuitive when the discriminatory covariate has support equal to the real line. We  discuss the case when $D_i$ is discrete in the remark below. 

\begin{remark}\label{remark:discrete}
    If $D_i$ has support that is not equal to the real line such as a discrete covariate, then the perturbation $D_{i,\ep}$ yields values outside the support of $D_i$.  In this scenario, we mollify $D_i$. To illustrate, let $D_i$ take values $t_1, \ldots, t_K$ with probabilities $p_1, \ldots, p_K$, $\sum_{k = 1}^K p_k = 1$, and for simplicity assume that $(\D_{-i}, \X)$ has support $\R^{n+m-1}$, where $\D_{-i}$ denotes the vector $\D$ deprived of its $i$-th component. First, we extend the aggregation function $h\colon \R^{n+m-1}\times \{t_1, \ldots, t_K\}$ to $\bar{h}\colon \R^{n+m}\to \R$ via
\begin{equation*}
    \bar{h}(\x, \d_{-i}, t):= 
    \begin{cases}
            h(\x, \d, t_1) \qquad &\text{if} \quad t \le t_1
            \\[0.5em]
            h(\x, \d, t_2) \qquad &\text{if} \quad t_1 < t \le t_2
            \\
            \qquad \vdots
            \\
            h(\x, \d, t_K) \qquad &\text{if} \quad t_{K-1}<t \,,
    \end{cases}
\end{equation*}
in which case $\bar{h}(\X, \D) = h(\X, \D)$ $\P$-a.s.. Second we mollify $D_i$. Often, the values $t_1, \ldots, t_K$ assigned to the categories of $D_i$ are chosen arbitrarily. Thus, instead of discrete values, we may define a rv $\hat{D}_i$ that is comonotonic to $D_i$ and has density
\begin{equation}\label{eq:pert-discrete-density}
    f_{\hat{D}_i}(x):= \frac{1}{\tau}\sum_{k = 1}^K\phi\left(\frac{x- t_k}{\tau}\right)\,,
\end{equation}
where $\phi$ denotes the standard normal density and  $\tau >0$. Thus, $\hat{D}_i$ is a rv generated by a kernel density estimation of $D_i$ with bandwidth $\tau$, where for simplicity of notation we suppress the dependence of $f_{\hat{D}_i}$ and $\hat{D}_i$ on $\tau$. Clearly, if we take the bandwidth $\tau$ to zero the rv $\hat{D}_i$ converges to $D_i$, that is $\lim_{\tau \downarrow 0} \hat{D}_i = D_i$ $\P$-a.s..

Thus, as $\bar{h}$ is defined on $\R^{n+m}$ the proportional perturbation applied to $\hat{D}_i$ and evaluated at $\bar{h}$, i.e. $\bar{h}(\X, \hat{D}_{i,\ep})$, is well-defined. Moreover, as long as the bandwidth $\tau$ is chosen small enough, $\P\big(\bar{h}(\X, \D) = \bar{h}(\X, \D_{-i}, \hat{D}_i)\big)$ is close to 1 and therefore, $\bar{h}(\X, \D_{-i}, \hat{D}_i)$ is a good approximation for $h(\X, \D)$. 
Hence, whenever $D_i$ is categorical, we propose to approximate $D_i$ with a continuously distributed rv $\hat{D}_i$ whose clusters correspond to the categories of $D_i$. For example a value of $\hat{D}_i$ that is close to $t_k$, corresponds to $D$ taking the value $t_k$. Thus, the proportional perturbation $\hat{D}_{i,\ep} = \hat{D}_i (1 + \ep)$ is well-defined and the distribution of $\hat{D}$ is modified in a non-trivial manner.
\vspace{1em}
\end{remark}

Next, we derive the representation of the sensitivity of the distortion risk principle which will be important for the exposition. The results hold under mild assumptions on the aggregate loss $Y$, and we further show below that the resulting formula for the sensitivity is a good approximation for the sensitivity when perturbing a discrete covariate. 

In order to guarantee the existence of the sensitivity of the distortion risk principle, we require the below assumption. For fixed $\ep > 0$, we denote by $F_{i, \ep}(\cdot):= \P(h(\X, \D_{i,\ep}) + \mfu \le \cdot ~|~ \X = \x)$ the conditional cdf of the distorted output and by $\Finv_{i, \ep}(\cdot)$ its quantile function.

\begin{assumption}
\label{asm: cont}
    Assume that the function $\ep \mapsto F_{i,\ep}(y)$ is continuously differentiable for all $y$ in the interior of $\supp{Y_\ep}$, and that $\ep \mapsto \Finv_{i,\ep}(u)$ is differentiable for all $u \in (0,1)$.
\end{assumption}

The next result, which is a consequence of \cite{PesentiMillossovichTsanakas2025}, provides a representation of the sensitivity of the distortion risk principle to a protected covariate.

\begin{lemma}\label{lemma:rep-sensitivity}

Let $F_\ep$ be such that Assumption \ref{asm: cont} is satisfied. Then, the sensitivity of $\rho(Y|_\x)$ to covariate $D_i$ has representation
\begin{align*}
    \partial_{D_i} \rho(Y|_\x) &= 
    \Ex\left[\Phi_i(\X, \D, U_{Y|_\x})\right]\,,
\end{align*}
where $\Phi_i(\x, \d, u) := d_i \;\partial_i h(\x, \d) \gamma(u)$.
\end{lemma}

Note that even if some of the covariates $(\X,\D)$ are discrete, the assumptions in \Cref{lemma:rep-sensitivity} can still be fulfilled as they pertain to continuity of the perturbed aggregate loss $Y$.

\begin{remark}
    We continue \Cref{remark:discrete}, the case when $D_i$ is discrete and we perturb the continuous rv $\hat{D}_i$ with density in \eqref{eq:pert-discrete-density} instead of $D_i$. In this case $\partial_{\hat{D}_i}\rho\big(\bar{h}(\X, \D_{-i}, \hat{D}_i)|_\x\big)$ is close (for small enough bandwidth) to $\partial_{D_i}\rho(Y_\x)$. Indeed, under mild integrability assumptions, we obtain (using \Cref{lemma:rep-sensitivity}) that
\begin{align*}
   \lim_{\tau \downarrow 0} \partial_{\hat{D}_i}\rho(Y|_\x)
    &=
    \lim_{\tau \downarrow 0} \; \Ex\big[ \hat{D}_i \;\partial_i \bar{h}(\X, \D_{-i}, \hat{D}_i) \gamma\big(F_{Y|_\x}(Y)\big)\big]
    \\
    &= 
    \Ex\left[\lim_{\tau \downarrow 0}  \; \hat{D}_i \;\partial_i \bar{h}(\X, \D_{-i}, \hat{D}_i) \gamma\big(F_{Y|_\x}(Y)\big)\right]
    \\
    &=
    \Ex\big[ D_i \;\partial_i h(\X, \D_{-i}, D_i) \gamma\big(F_{Y|_\x}(Y)\big)\big]
    \\
    &= 
    \partial_{D_i}\rho(Y|_\x)\,.
\end{align*}
Thus, even if $D_i$ is discrete, the formula in \Cref{lemma:rep-sensitivity} provides a good approximation to the sensitivity when perturbing $\hat{D}_i$, the mollified version of $D_i$. Thus, throughout we will utilise \Cref{lemma:rep-sensitivity} independent of the support of $D_i$.
\vspace{1em}
\end{remark}

\subsection{Discrimination-insensitive principle}

In this section we define a discrimination-insensitive probability measure such that the distortion risk principle under that measure has zero sensitivity, i.e. is insensitive, to protected covariates. That is, the pricing principle becomes insensitive to distributional perturbation of the protected covariate. We note that being insensitive is both a property of the distortion risk principle $\rho$ and the underlying data $(\X, \D)$ as the next example shows.

\begin{example}
   Consider a single permitted characteristic $X$, a single discriminatory covariate $D$, a linear aggregation function $h(x, d) = \beta_1 x + \beta_2 x d + \mfu$, $\mfu$ is an error term with mean zero, and the expected value, also called pure premium in the actuarial literature. Then, the sensitivity with respect to $D$ is 
\[
\partial_D \Ex[Y] = \Ex[\Phi(X, D, U_{Y|_\x})] = \beta_2\, x \;\Ex[D]\;,
\]
which is not zero unless $\Ex[D]$ or $x$ is zero. 
\end{example}

Thus, with the aim to modify premium principles to become insensitive to protected covariates, we propose to construct a probability measure $\Q$ that is ``closest" to the real world probability measure $\P$, and which has zero sensitivity to protected covariates.

In order to quantify the difference between probability measures, we employ the popular Kullback-Leibler (KL) divergence. For a probability measure $\Q$ on $(\Omega, \F)$, the KL divergence from $\Q $ to $\P$ is defined as
\begin{equation*}
    D_{KL}(\Q ~||~ \P) := \E\left[\frac{d\Q}{d\P} \log\left(\frac{d\Q}{d\P}\right)\right]\,,
\end{equation*}
if $\Q$ is absolutely continuous with respect to $\P$ and $+\infty$ otherwise.
We choose the KL divergence to quantify the deviation between probability measures as the KL divergence preserves the support of the covariates.  In the context of insurance and policyholder-provided information, many covariates are discrete or categorical. Classical examples include gender, age groups, and nationality.

Next, we define the optimal pricing measure as the closest probability measure to $\P$ in the KL divergence, under which the premium has zero sensitivity to all discriminating characteristics and maintains the same expectation conditional on $\x$.

\begin{definition}[Discrimination-insensitive pricing measure]
Let Assumption \ref{asm: cont} hold. Then, for some policyholder covariates $\x \in \supp{\X}$, the discrimination-insensitive measure $\Q_\x^*$ (if it exists) is the solution to
\begin{subequations}\label{eq: discrim_measure}
    \begin{align}
    \argmin_{\Q \ll \P} D_{KL}(\Q ~||~\P)
    \quad \text{subject to}\quad
    & \partial_{D_i} \rho^\Q(Y |_{\x}) = 0 \,,
    \quad \text{for all} \quad i \in \M\,, \quad\text{ and}\\
    & \Ex^\Q [Y] = \Ex^\P[Y].
    \label{eq:discrim_measure-sensitivity-constraint}
\end{align}
\end{subequations}
\end{definition}

We interpret the solution to \Cref{eq: discrim_measure} as the ``closest" probability measure to the real-world measure $\P$ that is insensitive to the discriminatory covariates and has the same expected losses as under $\P$.  The expectation constraint prevents cross-subsidy between subgroups.  With this discrimination-insensitive ``pricing'' measure, a financial company can calculate discrimination-insensitive prices as follows. 

\begin{definition}[Discrimination-insensitive principle]
Assume there exist a solution to Optimisation Problem \eqref{eq: discrim_measure}. Further, let $\x \in \supp{\X}$ be some policyholder covariates and $\Q^\star_\x$ be a discrimination-insensitive measure, i.e. a solution to \eqref{eq: discrim_measure}. Then the discrimination-insensitive pricing principle is given by
\begin{align*}
\rho^{\Q^\star_\x}(Y|_\x) &:= \int_0^1 \Finv_{Y|_\x}^{\Q^\star_\x}(u) \gamma(u) du
= \Ex^{\Q_x^\star}\left[Y \, \gamma(U^{\Q^\star_\x}_{Y|_\x}) \right]\,,
\end{align*}
where $\Finv^{\Q^\star_\x}_{Y|_\x}$ is the quantile function of $Y|_\x$ under $\Q^\star_\x$, and where $U^{\Q^*_\x}_{Y|_{\x}}:= F^{\Q^*_\x}_{Y|\x}(Y|_\x)$ is conditional on $\x$ a uniform rv under $\Q^\star_\x$.
\end{definition}

As we require the conditional expectation constraint, that is, the conditional expectations of the losses under the discrimination-insensitive measure equals those under $\P$, our setting is suitable when prices are calculated under the real world probability $\P$. This includes credit loans and insurance premia. Moreover, in this work, the insensitivity is with respect to each covariate individually.  That is, each $D_i$ is perturbed while $\D_{-i}$ and $\X$ are kept fixed, and as the perturbation $D_{i, \ep}$ is comonotonic to $D_i$, the copula of $(\X, \D_{-i}, D_i)$ and $(\X, \D_{-i}, D_{i,\ep})$ are the same. We leave the examining of indirect discrimination, as formalised in \cite{LindholmEtAl2022discrimfree}, for future research.

\section{Discrimination-insensitive pricing measure}
\label{sec: disc-ins-pricing-measure}

This section is devoted to solving optimisation problem \eqref{eq: discrim_measure} and deriving a representation of the discrimination-insensitive probability measure. Moreover, we provide conditions for existence of a solution to optimisation problem \eqref{eq: discrim_measure} and show that it is unique if it exists. 

Before stating the representation of the discrimination-insensitive probability measure, we first discuss the sensitivity of the distortion risk measure under an alternative probability measure $\Q \ll \P$.
By \Cref{lemma:rep-sensitivity} the sensitivity of the distortion risk principle under $\Q$ has representation
\begin{equation*}
    \partial_{D_i} \rho^\Q (Y|_\x) 
    =
    \Ex^\Q\left[ \Phi_i(\X, \D, U^\Q_{Y|_{\x}})\right] \,,
\end{equation*}
where $U^\Q_{Y|_{\x}}$ is under $\Q$ a uniform rv that is comonotonic to $Y|_\x$. Next, we rewrite the sensitivity constraint in \eqref{eq: discrim_measure} as 
\begin{equation}\label{eq:sens-q-unit}
    \partial_{D_i} \rho^\Q (Y|_\x) 
    =
    \Ex^\Q\left[ \Phi_i(\X, \D, V_{Y|_{\x}})\right] = 0
    \,, \quad 
    \text{subject to} \quad \Q(V_{Y|_{\x}} \le u) = u\,
\quad \text{for all } u \in (0,1)\,,
\end{equation}
where $V_{Y|_{\x}}$ is under $\P$ a uniform rv that is comonotonic to $Y|_\x$. The infinite number of constraints in \eqref{eq:sens-q-unit} enforces that $V_{Y|_{\x}}$ is also a uniform rv under $\Q$. Moreover, whenever $\Q$ and $\P$ are equivalent, we have that $V_{Y|_{\x}}$ is comonotonic to $Y|_\x$ under $\Q$. While we do not impose equivalence of the alternative probability measures in optimisation problem \eqref{eq: discrim_measure}, the solution to \eqref{eq: discrim_measure} will be equivalent to $\P$. Note that rewriting the sensitivity as in \eqref{eq:sens-q-unit}, both the sensitivity an the probability constraints  become linear constraints with respect to the RN derivative $\frac{d \Q}{d \P}$.

\begin{theorem}[Representation]
\label{thm:representation}

Let Assumption \ref{asm: cont} be satisfied. Denote a solution to optimisation problem \eqref{eq: discrim_measure} if it exists by $\Q_\x^{*}$. Then $\Q_\x^{*}$ is unique and its Radon-Nikodym (RN) derivative to $\P$ admits representation
\begin{align*}
    \frac{d\Q^{*}_\x}{d\P} &= \frac{1}{C_{\x}}\e^{-\bfeta_\x^{*\intercal} \Phi (\X, \D, U_{Y|_{\x}}) - \eta^*_{m+1} Y}  \;,
    \quad \text{where}
    \\[0.75em]
    C_{\x}&:=\E_{\x}\left[e^{-\bfeta_\x^{*\intercal} \Phi (\X, \D, U_{Y|\x}) - \eta^*_{m+1} Y   }~\big|~U_{Y|\x}\right]\,,
\end{align*}
where $\Phi(\x, \d, u) := (\Phi_1(\x, \d, u), \ldots, \Phi_m(\x, \d, u))$, $\bfeta^*_\x := (\eta^*_{\x,1}, \ldots, \eta^*_{\x,m})$ with $\eta^*_{\x,i} \in \R$ are the optimal Lagrange multipliers such that $\partial_{D_i} \rho^{\Q^*_\x}(Y |_{\x}) = 0$ is satified for all $i \in \M$, and $\eta^*_{x, m+1} \in \R$ is the optimal Lagrange multiplier such that $\Ex^\Q[Y] = \Ex^\P[Y]$ holds.
\end{theorem}

\textit{Proof} 
We let $\x \in \supp{\X}$ be fixed. For simplicity of notation, we assume that the joint density of $(\D, U_{Y|_{\x}})$ conditional on $\X$ exists. The arguments in this proof can be generalised to $(\D, U_{Y|_{\x}})$ having arbitrary distributions. Let $\Q \ll \P$ (possibly depending on $\x$) be a probability measure on $(\Omega, \F)$ and denote by $f(\cdot, \cdot)$ and $g(\cdot, \cdot)$ the density of $(\D, U_{Y|_{\x}})|_{\x}$, under $\P$ and $\Q$, respectively. For simplicity, where we omit the dependence on $\x$, and further denote the support of $(\D, U_{Y|_{\x}})|_\x$ by $\A\subseteq \R^{m+1}$.

Next, we rewrite the sensitivity constraints of \eqref{eq:sens-q-unit} in terms of the $\Q$ density $g$ and obtain
\begin{equation*}
    \partial_{D_i} \rho^\Q (Y|_\x) 
    =
    \int_{\A} g(\d, u) \Phi_i(\x, \d, u) \, d(\d, u)\,,
    \quad \text{subject to} \quad \mathfrak{g}(u):= \int_{\supp{D}}g(\d, u) d \d = 1 \,
\end{equation*}
for all $u \in (0,1)$. At this stage, we do not make the assumption that the probability measure $\Q$ is equivalent to $\P$. From the representation of the optimal RN derivative, we obtain that the optimal probability measure is equivalent to $\P$.

Rewriting Optimisation Problem \eqref{eq: discrim_measure} as an optimisation problem over densities, we obtain
\begin{align*}
\label{eq: discrim_opt_integral}
    \min_{g: \R^{m+1} \to\R^{m+1}} \int_{\A} g(\d, u) \log\left(\frac{g(\d, u)}{f(\d, u)}\right) d(\d, u)
    \quad &\text{subject to}\\[0.5em]
    &\int_{\A} g(\d, u) \Phi_i(\x, \d, u) \, d(\d, u) = 0 \quad \text{for all} \quad i \in \M \,, 
    \\[0.5em]
    & \int_\A h(\x, \d)\; g(\d, u) d(\d,u) = \int_\A h(\x, \d)\; f(\d, u) d(\d,u)\;,
    \\[0.5em]
    &\int_{\A} g(\d, u) d(\d, u) = 1\;, \\[0.5em]
    &g( \d, u) \geq 0 \quad \text{for all }\quad ( \d, u) \in \A \;,
    \\[0.5em]
    & \mathfrak{g}(u)  = 1 \quad \text{for all } u \in (0,1) \;.
\end{align*}
The first constraint corresponds to the insensitivity constraint and the second to the conditional expectation constraint, see \Cref{eq: discrim_measure}. The third and forth correspond to $g$ being a conditional density, i.e. integrating to 1 and being non-negative on its support. The final constraint is such that $\mathfrak{g}(\cdot)$ is a uniform density, which is equivalent to $\Q(U_{Y|_\x}\le u ) = u$ for all $u \in (0,1)$. 
Hence, the corresponding Lagrangian with Lagrange multipliers $\bfeta := \big(\eta_1, \ldots,  \eta_{m+2}, \eta_{m+3}(\d,u), \eta_{m+4}(u)\big)$, where we omit the dependence of $\bfeta$ on $\x$, is
\begin{align*}
    \mathcal{L}(g; \bfeta) &= \int_\A \Bigg( g(\d, u) \log\left(\frac{g(\d, u)}{f(\d, u)}\right) + g(\d, u) \sum_{i \in \M} \eta_i \Phi_i (\x, \d, u) 
    +
    \eta_{m+1} h(\x, \d) \; g(\d, u)+ \eta_{m+2} \; g(\d,u)\\
    &\qquad  - \eta_{m+3}(\d,u) g(\d, u) + \eta_{m+4}(u)\left(g(\d, u)-1\right) \Bigg)\; d(\d, u) \\
    &\qquad- \eta_{m+1} \int_A h(\x, \d)\; f(\d, u) d(\d, u) - \eta_{m+2}\;,
\end{align*}
and the Lagrange multipliers  $\bfeta $, 
satisfy $\eta_i \in \R$, $i \in \M \cup \{m+ 1, m+2\}$, $\eta_{m+3}(d,u) \geq 0$ for all $(\d,u) \in \A$, and $\eta_{m+4}(u) \in \R$ for all $u \in (0,1)$. The associated Euler-Lagrange equation is
\begin{equation}
    \log\left(\frac{g(\d, u)}{f(\d, u)}\right) + 1 + \sum_{i \in \M} \eta_i \Phi_i (\x, \d, u) + \eta_{m+1} h(\x, \d)  + \eta_{m+2} - \eta_{m+3}(\d,u) + \eta_{m+4}(u) = 0 \;.
\end{equation}

Imposing $g(\d, u) \geq 0$ for all $(\d, u) \in \A $ yields $\eta_{m+3}(\d, u) = 0$ for all $(\d, u) \in \A$, as 
\begin{equation}\label{eq:frac-g-f}
    \frac{g(\d, u)}{f(\d, u)} = \exp\left\{-\sum_{i \in \M} \eta_i \Phi_i (\x, \d, u) - \eta_{m+1} h(\x, \d) - \eta_{m+2}+   \eta_{m+4}(u) \right\}\;, 
\end{equation}
where we reparametrised $\eta_{m+2}:= \eta_{m+2}+1$.
Next, imposing the constraint that $\mathfrak{g}(u) = 1$, we obtain
\begin{align*}
    1 &= \int_{\supp{\D}}g(\d, u)d\d  
    \\
    &= 
    \int_{\supp{\D}}f(\d, u)\exp\left\{-\sum_{i \in \M} \eta_i \Phi_i (\x, \d, u) - \eta_{m+1} h(\x, \d)  - \eta_{m+2}  +\eta_{m+4}(u) \right\}d\d
    \\&= 
    e^{- \eta_{m+2} }\, e^{ \eta_{m+4}(u) }\,C_{\x, u} \,,
\end{align*}
where we define $C_{\x, u}:= \int_{\supp{\D}}f(\d, u)e^{-\sum_{i \in \M} \eta_i \Phi_i (\x, \d, u) - \eta_{m+1} h(\x, \d)}d\d$. Thus, $e^{ \eta_{m+4}(u) } =  e^{ \eta_{m+2}} / C_{\x, u}$ and \eqref{eq:frac-g-f} simplifies to
\begin{equation*}
        \frac{g(\d, u)}{f(\d, u)} = \frac{1}{C_{\x,u}}\exp\left\{-\sum_{i \in \M} \eta_i \Phi_i (\x, \d, u) - \eta_{m+1} h(\x, \d) \right\}\;. 
\end{equation*}
We note that $\eta_{m+2}$ cancels out and that $g(\d, u)$ integrates to one. Indeed,  
\begin{align*}
    \int_{\A} g(\d, u) d(\d, u)
    &=\int_0^1\int_{\supp{\D}}\frac{1}{C_{\x,u}} \e^{-\sum_{i \in \M} \eta_i \Phi_i (\x, \d, u) - \eta_{m+1} h(\x, \d) }  f(\d, u) d(\d,u)
    \\
    &=\int_0^1 \frac{1}{C_{\x,u}}\int_{\supp{\D}}\e^{-\sum_{i \in \M} \eta_i \Phi_i (\x, \d, u) - \eta_{m+1} h(\x, \d) } f(\d, u)d\d  \, du
    \\
    &= 1\,.
\end{align*}
Thus, the  Radon-Nikodym derivative becomes
\begin{equation*}
\frac{d\Q^\bfeta_\x}{d\P} = \frac{g(\D, U_{Y|_\x})}{f(\D, U_{Y|_\x})}=\frac{1}{C_{\x}}\e^{-\sum_{i \in \M} \eta_i \Phi_i (\X, \D, U_{Y|_\x}) - \eta_{m+1} Y} \;,
\end{equation*}
where the random variable $C_{\x}$ is define in the statement.
We observe that since $\P$ and $\Q_\x^\bfeta$ are equivalent, we have that $(Y, U_{Y|_\x})$ is comonotone under both $\P$ and $\Q$. 

For uniqueness, recall that  the KL-divergence is strictly convex in $\frac{d \Q}{d \P}$. Moreover, the sensitivity and the conditional expectation constraints are linear in $\frac{d \Q}{d \P}$. Additionally, the constraint that $U_{Y|_\x}$  is uniform under $\Q$ is also linear in the RN derivative. Indeed for all $u \in (0,1)$, we have
that $\Q(U_{Y|_\x} \le u) =\Ex[\Id_{\{U_{Y|_\x} \le u\}} \frac{d \Q}{d \P}]$.
Thus, if a solution exists, it must be unique. Moreover, assuming that a solution to Optimisation Problem \eqref{eq: discrim_measure} exists is equivalent to the existence of optimal Lagrange multipliers $\bfeta^*$ such that the constraints are fulfilled. Thus, the RN derivative is $\frac{d\Q^{\bfeta^{*}}_\x}{d\P}$, and for simplicity of exposition we write $\Q_\x^*:= \Q_\x^{\bfeta*}$.
\hfill \Halmos

Recall that $U_{Y|_\x}$ is under $\P$ a uniform rv that is comonotonic to $Y|_\x$. Furthermore, by construction of $\Q^*_\x$, we also have that $U_{Y|_\x}$ is under $\Q^*_\x$ a uniform rv. Moreover, since $\P$ and $\Q^*_\x$ are equivalent, the rv $U_{Y|_\x}$ is comonotonic to $Y|_\x$. Thus, rewriting the sensitivity constraint as \eqref{eq:sens-q-unit} is justified.

Of particular interest is that the pricing measure depends on all covariates, the aggregation function, and the pricing principle, as all these components may propagate or even amplify discrimination.
 
For the existence of the discrimination-insensitive measure, we require the definition of the conditional cumulant generating function of $\W_{\x,u} := \big(\Phi(\X, \D, U_{Y|\x}), \, Y - \Ex[Y]\big)^\intercal$, given $\X = \x$ and $U_{Y|\x} = u$, which we denote by 
\begin{equation*}
    \mathbf{K}_{\W|_{\x, u}}(\t) = \log\left(\E_{\x, u} \left[ \e^{\t^\intercal \W} \right]\right)\,,
    \quad \text{for} \quad \t \in \R^{m+1}\,,
\end{equation*}
where we denote by $\E_{\x,u}[\cdot]$ the expectation conditional on $\X = \x$ and $U_{Y|\x} = u$. 

\begin{theorem}[Existence]\label{thm:existence}
Let Assumption \ref{asm: cont} be satisfied and let $\W|_{\x,u}$ be non-degenerate under $\P$ for all $\x \in \supp{\X}$ and $u \in (0,1)$. Assume that for each $(\x, u)$, $|\nabla_\bt\,  \mathbf{K}_{\W |_{\x, u}}(\bt) |$ is dominated by an integrable function, and there exists a solution to the system of equations in $\bfeta$
\begin{equation}
\label{eq:sol-eta}
    \int_0^1\nabla_\bt\,  \mathbf{K}_{\W |_{\x, u}}(\bt) \big\vert_{\bt = -\bfeta}\, du = 0 \,,
\end{equation}

where $\nabla_{\bt}$ denotes the gradient with respect to $\bt$. Then 
\begin{enumerate}[label = $\roman*)$]
    \item the solution to \eqref{eq:sol-eta}, denoted by $\bfeta^\star$, is unique,

    \item $\bfeta^\star$ is the optimal Lagrange multipliers of the RN derivative in  \Cref{thm:representation}.
\end{enumerate}
\end{theorem}

\textit{Proof} 
For part $i)$, we first show that the joint cgf of a non-degenerate random vector $\Z:= (Z_1, \ldots, Z_n)$, where $Z_i \in \Lp$, $i \in\mN$, $\mathbf{K}_{\Z}(\t) = \log(\E[\e^{\t^\intercal \Z}])$ is strictly convex in $\t \in \R^n$. 
The first partial derivative of the cgf with respect to the $k$-th element of $\t$, $k \in \mN$, is
\begin{equation*}
    \partial_{t_k} \mathbf{K}_\Z(\t) = \partial_{t_k}\log\left(\E[\e^{\t^\intercal \Z}] \right)
    = \frac{\E[Z_k \e^{\t^\intercal \Z}]}{\E[\e^{\t^\intercal \Z}]} = \E^{\Tilde{\Q}}[Z_k]\;,
\end{equation*}
where $\Tilde{\Q}$ is defined via the RN derivative $\frac{d\Tilde{\Q}}{d \P} := \frac{\e^{\t^\intercal \Z}}{\E[\e^{\t^\intercal \Z}]}$. Then, the second partial derivative now with respect to the $k$-element of $\t$, $k \in \mN$, is
\begin{align}
   \partial_{t_j} \partial_{t_k} \mathbf{K}_\Z(\t) &= \E^{\Tilde{\Q}}[Z_j Z_k] - \E^{\Tilde{\Q}}[Z_j]\;\E^{\Tilde{\Q}}[Z_k] = \Cov^{\Tilde{\Q}}(Z_j, Z_k) \label{eq:cov}.
\end{align}
For \Cref{eq:cov} to be positive for all $j, k \in \mN$, the covariance matrix $\Cov(\Z):= \{\Cov^{\Tilde{\Q}}(Z_j, Z_k) \}_{j,k \in \mN}$ must be a positive definite matrix, which is true if and only if $\Z$ non-degenerate under $\Tilde{\Q}$. Moreover, $\Z$ being non-degenerate under $\Tilde{\Q}$ is equivalent to $\Z$ being non-degenerate under $\P$. Thus, as $(\x, u) \in \supp{\X}\times (0,1)$ the random vector $\W|_{\x, u}$ is non-degenerate, we have that $\bK_{\W|_{\x, u}}(\bt)$ is strictly convex in $\bt$. Moreover, for all $j,k \in \{1, \ldots, m+1\}$, we have by the dominated convergence theorem that
\begin{equation*}
   \partial_{t_j} \partial_{t_k} \int_0^1\bK_{\W|_{\x, u}}(\bt)\, du
   =
   \int_0^1\partial_{t_j} \partial_{t_k}\bK_{\W|_{\x, u}}(\bt)\, du >0 \,,
\end{equation*}
where the last inequality follows by \eqref{eq:cov}. 

For part $ii)$, we write with slight abuse of notation $\bfeta^* = (\eta_1^*, \ldots, \eta_m^*)$. The sensitivity constraints of Optimisation Problem \ref{eq: discrim_measure} vanishes for all $i \in \M$ if and only if 
\begin{align}
    \partial_{D_i} \rho^{\Q^*_\x}(Y|_\x)  
    &=
    \Ex\left[\Phi_i(\X,\D, U_{Y|_\x}) \, \frac{d\Q^{*}_\x}{d\P}\right]
    \nonumber
    \\[0.5em]
    &=
        \Ex\left[\Phi_i(\X,\D, U_{Y|_\x}) \, \frac{1}{C_{\x}}\e^{- \bfeta_\x^{*\intercal} \Phi (\X, \D, U_{Y|_{\x}}) - \eta^*_{m+1} Y}\right]
   \notag\\[0.5em]
    &=
    \int_0^1\frac{1}{C_{\x, u}}\E_{\x,u}\left[\Phi_i(\X,\D, U_{Y|_{\x}}) \, \e^{-\bfeta_\x^{*\intercal} \Phi (\X, \D, U_{Y|_{\x}}) - \eta^*_{m+1} Y}\right]\, du
     \label{eq:cgf-constraints-sens}\\[0.5em]
     &=0\;,
    \notag
\end{align}
where $C_{\x,u}$ is defined in the proof of \Cref{thm:representation}.
The conditional expectation constraint, we rewrite similarly,  
\begin{align}
    \Ex^{\Q^{*}_\x}[Y - \Ex[Y]] &= \Ex \left[\left(Y-\Ex[Y]\right) \, \frac{d\Q^{*}_\x}{d\P}\right]
    \nonumber\\[0.5em]
        &=\Ex\left[\left(Y - \Ex[Y]\right)\frac{1}{C_{\x}}\e^{-\bfeta_\x^{*\,\intercal} \Phi(\X,\D, U_{Y|_\x}) - \eta_{m+1}^* Y}\right] 
        \notag\\[0.5em]
        &= \int_0^1\frac{1}{C_{\x, u}}\E_{\x, u}\left[\left(Y - \Ex[Y]\right)\e^{-\bfeta_\x^{*\,\intercal} \Phi(\X,\D, U_{Y|_\x}) - \eta_{m+1}^* Y}\right] du 
         \label{eq:cgf-constraints-expec}\\[0.5em]
        &= 0\;.\notag
\end{align}
Hence, for each $i \in \bar{\M} =:\M \bigcup \{m+1\}$, \eqref{eq:cgf-constraints-sens} and \eqref{eq:cgf-constraints-expec} can be written as 
\begin{align}
   \int_0^1\tfrac{\partial}{\partial_{t_i}} \; \bK_{\W|_{\x, u}} (\t) \vert_{\t = -\bfeta} \, du &= 0\,.
\end{align}

Hence the solution to \eqref{eq:sol-eta} are indeed the Lagrange multipliers for which the constraint of Optimisation Problem \eqref{eq: discrim_measure} are satisfied. 
\hfill \Halmos

\begin{remark}
    In situations where the expectation constraint in Optimisation Problem \eqref{eq: discrim_measure} is not desired, all results apply by setting $\eta_{\x, ,+1}^*$ in \Cref{thm:representation} equal to zero. We also refer to \Cref{lemma: marginal-measures}, which gives the solution to Optimisation Problem \eqref{eq: discrim_measure} with only the sensitivity constraint.
\end{remark}
\vspace{1em}

When the risk measure is the expected value, i.e. $\gamma(u) \equiv 1$, then the formulas simplify. The proofs follow along the lines of the proofs of \Cref{thm:representation} and \Cref{thm:existence}, by noting that if $\gamma(u) \equiv 1$, then the constraint $\Q(U_{Y|_\x} \le u)  = u$, $u \in (0,1)$, is redundant.
\begin{corollary}\label{cor:representation-mean}
Let Assumption \ref{asm: cont} be satisfied and $\gamma(u) \equiv 1$. Denote a solution to \Cref{eq: discrim_measure} if it exists by $\Q_\x^{*}$. Then $\Q_\x^{*}$ is unique and its Radon-Nikodym (RN) derivative to $\P$ admits representation
\begin{align*}
    \frac{d\Q^{*}_\x}{d\P} &= \frac{\e^{-\bfeta_\x^{*\,\intercal} \Phi(\X,\D, U_{Y|_\x}) - \eta_{m+1}^* Y }}{\E_\x\left[\e^{-\bfeta_\x^{*\,\intercal} \Phi(\X,\D, U_{Y|_\x}) - \eta_{m+1}^* Y }\right]} \;,
\end{align*}
where $\bfeta^*_\x  \in \R^m$ are the Lagrange multipliers such that $\partial_{D_i} \rho^{\Q^*_\x}(Y |_{\x}) = 0$ for all $i \in \M$, and $\eta^*_{x, m+1} \in \R$ is the Lagrange multiplier such that $\Ex^\Q[Y] = \Ex^\P[Y]$ holds.
\end{corollary} 
Continuing in the case of the expected value, the existence simplifies to finding the root of the derivative of the cgf of $\W|_{\x} := \big(\Phi(\X, \D, 1), \, Y - \Ex[Y]\big)^\intercal$.
\begin{corollary}
    Let Assumption \ref{asm: cont} be satisfied,  and $\gamma(u) \equiv 1$, and let $\W|_{\x} := \big(\Phi(\X, \D, 1), \, Y - \Ex[Y]\big)^\intercal$ given $\X = \x$, be non-degenerate under $\P$ for all $\x \in \supp{\X}$. Assume there exists a solution to the system of equations in $\bfeta$
\begin{equation}
\label{eq:sol-eta-mean}
    \nabla_\bt\,  \mathbf{K}_{\W |_{\x}}(\bt) \big\vert_{\bt = -\bfeta} = 0 \,.
\end{equation}
Then 
\begin{enumerate}[label = $\roman*)$]
    \item the solution to \eqref{eq:sol-eta-mean}, denoted by $\bfeta^\star$, is unique,
    \item $\bfeta^\star$ is the optimal Lagrange multipliers of the RN derivative in  \Cref{cor:representation-mean}.
\end{enumerate}
\end{corollary}

There are situations when the system of equations \eqref{eq:sol-eta} does not have a solution. First, we require that the cgf of $\W |_{\x, u}$ exists. In order to preserve generality of this work, we do not restrict the rvs, for example, to be essentially bounded or have subexponential decay; in which cases cgfs always exist. Even if the cgfs exists, the system of equations \eqref{eq:sol-eta} may not have a solution, e.g., if the equations are affine transformation of one another. To illustrate, 
    consider an in insurance scenario in which the premium is calculated using the expected value (i.e. $\gamma(\cdot) \equiv 1$), and the policy holder claim follows a linear regression with aggregation function $h(x,d) = \beta_x x + \beta d + \beta_0 + \mfu$, where $(\beta_x$, $\beta_d$, $\beta_0) \in \R^3$ are coefficients corresponding to the permitted covariate, protected covariate and intercept, respectively, and $\mfu$ is a random error term with mean zero.
    Then, \eqref{eq:sol-eta} becomes
\begin{equation*}
    \e^{\eta_k \beta_x x} \left(\beta_x x - \E_x[Y]\right) = 0\;,
    \quad k = 1,2,
\end{equation*}
which can only be attained as $\eta_k \to -\infty$, $k = 1,2$. As a consequence, $\frac{d\Q}{d\P}$ is not a valid RN derivative, and thus no solution exists. 

We next discuss a simple example illustrating the discrimination-insensitive premia.

\begin{example}[ES risk-loaded expected value premium in a regression setting]
\label{ex: ES-risk-load}
We consider a scenario in which the premium is calculated using the expected value with an Expected Shortfall (ES) loading, that is,  $\gamma(u) = (1 + \frac{c}{1-\alpha}\Id_{\{u > \alpha\}} )$, for $u \in [0,1]$, $\alpha \in (0,1)$ and risk loading parameter $c >0$, yielding a weighted combination of the mean and ES, i.e., $\rho(Y) = \E[Y] + c\;\text{ES}_{\alpha}(X)$. While distortion risk measures are defined for weight functions $\gamma$ that integrate to 1, i.e. $\int_0^1 \gamma(u) du  = 1$, all statements in this manuscript also hold for a weight function that integrates to a positive constant. The policy holder claim follows a linear regression, that is
\[
    h(x,d) = \beta_x x + \beta d  + \beta_0 + \mfu\,,
\]
where $(\beta_x$, $\beta_d, \beta_0) \in \R^3$ are coefficients corresponding to the permitted covariate, protected covariate and intercept, respectively, and $\mfu \sim N(0, \sigma^2)$ is a normal error term with mean zero. The coefficients $\beta_x$, $\beta_d$ and $\beta_0$ are assumed to be known, e.g., through estimation. In this example, 
\begin{equation*}
    \Phi(x, d, u) = d \;\beta \; \left(1 + \frac{c}{1-\alpha} \Id_{\{u> \alpha\}}\right)\;.
\end{equation*}

The sensitivity under $\P$ to the discriminatory covariate becomes
\begin{align*}
    \partial_D \E_x\left[h(x, D) \gamma(U_{Y|_x})\right] 
    &= \E_x\left[\Phi(x, D, U_{Y|_x})\right]
    = \beta \;\E_x\left[D\right] + \frac{\beta \;c}{(1-\alpha)}\E_x\left[D \;\Id_{\{U_{Y|_{x}} > \alpha\}}\right]\;.
\end{align*}

Then by \Cref{thm:representation}, for each realisation of the permitted covariate $\x$, the discrimination-insensitive measures $\Q^*_x$ is characterised by

\begin{equation*}
\frac{d \Q^*_x}{d \P} = \frac{\e^{-\eta_1^* \Phi(X, D, U_{Y|_x}) - \eta_2^* Y}}{\E_{x}[\e^{-\eta_1^* \Phi(X, D, U_{Y|_\x}) - \eta_2^* Y} | U_{Y|_x}] ~}\;,
\end{equation*}

where $\eta^*_1$ and $\eta^*_2$ are the solutions (in $\eta_1$ and $\eta_2$) to the system of equations
\begin{subequations}\label{eq:constraints}
\begin{align}
    \int_0^1\frac{\E_{x,u}\left [\Phi(X, D, U_{Y|_x}) \e^{-\eta_1 \Phi(X, D, U_{Y|_x}) - \eta_2 Y}\right]}{\E_{x,u}\left [\e^{-\eta_1 \Phi(X, D, U_{Y|_x}) - \eta_2 Y}\right]}\; du 
    &= 0 \quad \text{and}
    \\[1em]
    \int_0^1\frac{\E_{x,u}\left[ \left(Y - \E_x[Y]\right) \e^{-\eta_1 \Phi(X, D, U_{Y|_x}) - \eta_2 Y}\right]}{\E_{x,u}\left[ \e^{-\eta_1 \Phi(X, D, U_{Y|_x}) - \eta_2 Y}\right]}\;du  &= 0 \,, 
\end{align}
\end{subequations}
where $\E_{x,u}[\cdot ]$ denotes the expectation conditional on $X = x$ and $|U_{Y|_x} = u$.

Collecting, our discrimination-insensitive premium becomes
\begin{equation*}
    \rho^{\Q^*_x}(Y |_{x}) 
    = \Ex\left[\left(\beta_x X + \beta D + \beta_0 + \mfu \right)\;\gamma\left(U_{Y|_x}\right) \frac{d\Q_x^*}{d\P}\right] \;.
\end{equation*}

Recall that by construction of $\Q_\x^*$, the rv $U_{Y|_x}$ is under $\Q_\x^*$ (and also under $\P$) standard uniformly distributed.  
In contrast, the unaware, the best-estimate, and the discrimination-free premia\footnote{We refer the reader to \cite{LindholmEtAl2022discrimfree}, who introduced the unaware and discrimination-free prices under the expected value principle.  Here, we generalise their definitions to distortion risk premia in order to compare with our methodology.} under $\P$ are respectively
\begin{align*}
 \rho(Y |_{x}) &= \Ex\left[Y \; \gamma\left(U_{Y|_x}\right)\right]\;, 
 \\
 \rho(Y |_{x, d})&= \E\left[Y \; \gamma\left(U_{Y|_x}\right)|X = x,D = d\right] \;,  \qquad \text{and} \\
 \rho_{\text{d.f.}}(Y) &= \sum_{w \in \supp{D}} \P(D = w | X = x)\rho(Y |_{x, w})\;.
\end{align*}

For the numerical implementation, we choose $h(x, d) = \frac{1}{2} x + \frac{1}{4}d + \mfu$, where $\mfu \sim \mathcal{N}\big(0, \big(\frac{1}{2}\big)^2\big)$, and let the discriminatory covariate $D$ represent ``gender". Specifically, $\P(D = -1) = 1 - \P(D = 1) =  \frac{2}{5}$,  where $-1$ represents ``female" and $+1$ represents ``male". The non-discriminatory covariate, or true ``riskiness", $X$, is modelled by
\begin{equation*}
    X \sim 
    \begin{cases}
  \;\mathcal{LN}(2, \frac{1}{3})  & \text{when } D = -1 \,,\\
  \;\mathcal{LN}(\frac{3}{2}, \frac{1}{2}) & \text{when } D = +1\,,
\end{cases}
\end{equation*}
where $\mathcal{LN}$ is the log-normal distribution with parameters on the log scale. This represents the scenario, where a riskier, in both mean and standard deviation, population is a smaller proportion of an insured pool. We consider the ES with tolerance level $\alpha = 0.9$ and risk-loading parameter $c = 0.2$.

In this example, the conditional distribution of $D$ given $X = x$ and of $Y$ given $X = x$ can be analytically calculated.\footnote{The conditional distribution of $D$ given $X = x$ can be easily calculated using Bayes' rule. The distribution of $Y$ given $X = x$ is 
\begin{align*}
     F_{Y|_x}(y) = \P(D = -1 ~|~ X = x) \Psi\left(\frac{y- \beta_x x + \beta}{\sigma}\right) + \P(D = 1 ~|~ X = x) \Psi\left(\frac{y- \beta_x x - \beta}{\sigma}\right)\;,
\end{align*}
where $\Psi$ is the cdf of a standard normal distribution.}
To find the discrimination-insensitive premium, a root solver procedure, and $1,000,000$ realisations of losses are used to find the optimal Lagrange parameters $\eta_1^*$ and $\eta_2^*$ for each $x$ using \eqref{eq:constraints}. This is repeated $50$ times and averaged to find the discrimination-insensitive premium, with non-convergent solutions being rejected.
\begin{figure}[h]
    \centering
    \includegraphics[width = 0.6\textwidth]{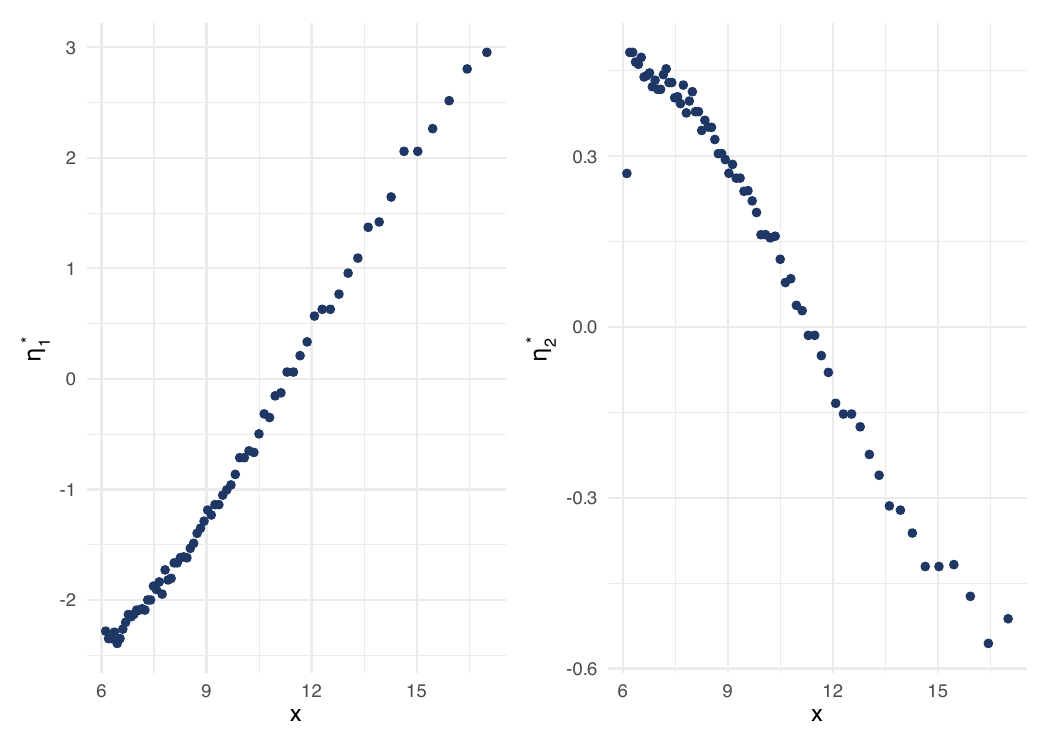}
    \caption{Estimated values of $\eta^*_1$ and $\eta^*_2$ for riskiness $x$.}
    \label{fig:risk-loading-etas}
\end{figure}
In \Cref{fig:risk-loading-etas}, the optimal Lagrange multipliers of one realisation of $X$ are plotted, showing smooth behaviour for both Lagrange multipliers.

\begin{figure}[h]
    \centering
    \includegraphics[width = 0.6\textwidth]{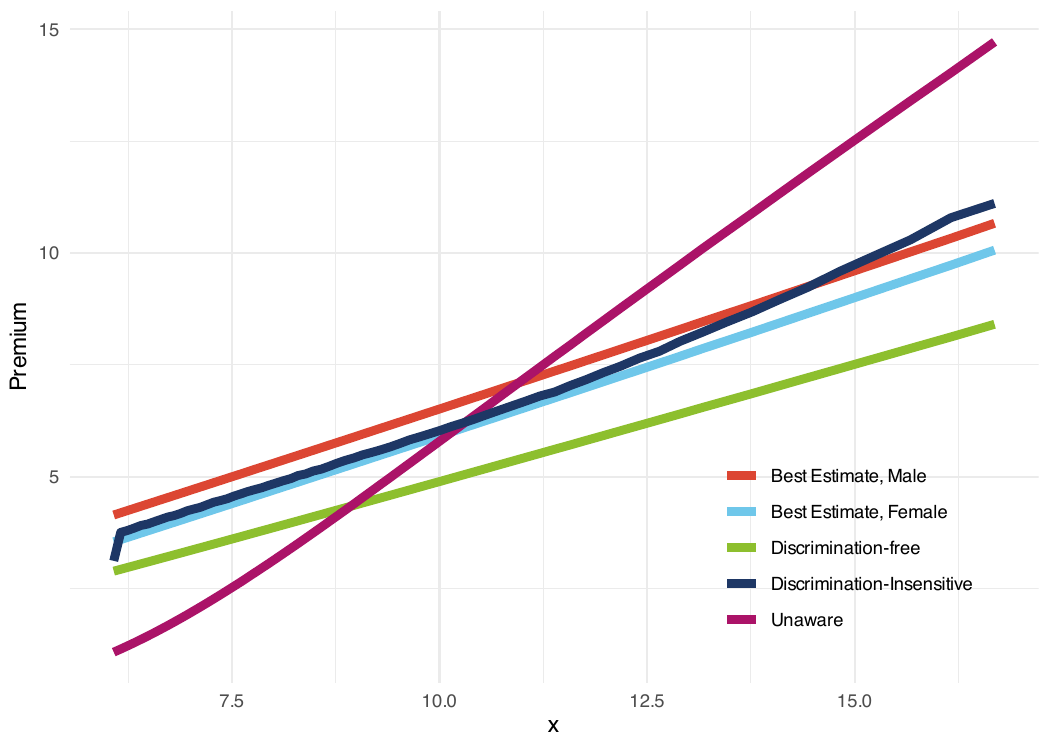}
    \caption{Best-estimate male (red) and female (light blue), discrimination-free (green), unaware (fuchsia), and discrimination-insensitive (dark blue) premia for riskiness $x$.}
    \label{fig:risk-loading-premia}
\end{figure}

\Cref{fig:risk-loading-premia} provides a visual comparison of the different premia. As expected, the best-estimate price of the female population is less than the best-estimate price for the male population. We observe that both the discrimination-free and the discrimination-insensitive premia can be seen as an interpolation of the two best-estimate premia. The discrimination-insensitive premium penalizes riskier behaviour, in that as the policyholder riskiness increases, the premium tends towards and exceeds that of the best-estimate of the riskier subpopulation, the ``males".
    
\end{example}
\section{Barycentre approach}
\label{sec: barycentre}

In situations when the discrimination-insensitive measure does not exist, for example when too many protected covariates are considered, we suggest to utilise an approximate discrimination-insensitive measure constructed via a two step procedure. First, we find marginal discrimination-insensitive measures for each discriminatory covariate separately, and second, reconcile these into a pricing measure via a KL barycentre approach -- that is, we find the probability measure which is closest to all marginal  discrimination-insensitive measures. While the barycentre approach is not truly discrimination-insensitive as defined in \eqref{eq:discrim_measure-sensitivity-constraint} i.e. that the sensitivity to each covariate vanishes, it is approximately discrimination-insensitive since it is closest (in the KL sense) to all of the marginal  discrimination-insensitive measures simultaneously.

In the literature, asymmetric and symmetric KL centroids or barycentres have been studied, with analytical solutions being provided when conditions permit, for example when considering probability mass or density functions, \cite{Veldhuis2002, NielsenNock2009}. These barycentres can be used to determine clusters of probability distributions \cite{BanerjeeEtAl2005}. Centroids are also used as a technique of model merging, sometimes called the method-of-experts or model fusion. Some references for this include \cite{Pelletier2005}, \cite{ClaiciEtAl2020}, and \cite{JaimungalPesenti2024}.
Here, we solve the KL barycentre problem in the probability measure setting, a generalisation of existing results and of interest of its own, in \Cref{thm:barycentre}. Furthermore we solve a constrained barycentre problem relevant to our application in \Cref{thm:constrained-barycentre}. 

As mentioned above, our barycentre approach includes two steps. First calculating marginal discrimination-insensitive measures for each discriminatory covariate separately. Different to the discrimination-insensitive measure defined in \eqref{eq: discrim_measure}, we do not impose the expected value constraint. Instead, the expected value constraint is imposed in the second step, when calculating the constraint barycentre measure. The reason for the expectation constraint is to ensure that the expected losses under $\P$ are equal to the expected losses under the pricing measure, thus this constraint should be imposed on the barycentre measure.

For the first step, we denote by $\Q_i$, $i \in \M$, (if they exist) the $i$-th \textit{marginal discrimination-insensitive measure} as the solution to 
\begin{equation}
     \argmin_{\Q \ll \P} D_{KL}(\Q ~||~\P)
    \quad \text{subject to}\quad
     \partial_{D_i} \rho^\Q(Y |_{\x}) = 0 \;.
     \label{eq: i-problem}
\end{equation}
Thus, $\Q_i$ is the measure that is discrimination-insensitive to the $i$th discriminatory covariate. In step two, the marginal discrimination-insensitive measure are reconciled into a constrained barycentre measure. The existence and uniqueness of $\Q_i$ is guaranteed by the following lemma.

\begin{lemma}
\label{lemma: marginal-measures}
Let the assumption of \Cref{thm:existence} be satisfied for $\W|_{\x,u} = \Phi_i(\X, \D, U_{Y|\x})$. Then for each $i\in \M$, there exist a unique solution of Optimisation Problem \eqref{eq: i-problem}, denoted by $\Q_i$. Moreover, for each $i \in \M$ its RN derivative to $\P$ is given by 
    \begin{equation*}
        \frac{d\Q_i}{d\P} = \frac{\e^{-\eta_i^{*} \Phi_i(\X, \D, U_{Y|_{\x}})} }{\E_{\x}\left[e^{-\eta_i^{*} \Phi_i(\X, \D, U_{Y|\x}) }~\big|~U_{Y|\x}\right]}\,,
    \end{equation*}
    where $\eta_i \in \R$ is the optimal Lagrange multiplier satisfying $\partial_{D_i} \rho^{\Q_i} (Y|_\x) = 0$.
\end{lemma}

\textit{Proof.} Representation of the RN derivative follows from \Cref{thm:representation} in the case where only the $i$th discriminatory covariate is considered, and the Lagrange parameter corresponding to the equal expectation constraint is $0$. Uniqueness follows similarly from \Cref{thm:existence}.  Moreover, the Lagrange multiplier $\eta_i$ is unique and always exists as a univariate cgf is strictly convex and differentiable. 

\hfill \Halmos

In the second step, we merge the marginally discrimination-insensitive measures into a unified model by using a weighted barycentre technique. We further require that the pricing measure preserves the reference conditional expectation of the losses with respect to the permitted covariates.  This is done to avoid cross-subsidies. That is, we define the constrained barycentre measure as the solution to optimisation problem
\begin{equation}
     \argmin_{\Q \ll \P} \sum_{i \in \M} \pi_i\,  D_{KL}(\Q ~||~\Q_i) \quad \text{subject to} \quad \Ex^{\Q}[Y] = \Ex^\P[Y]
     \label{eq: barycentre-problem}
\end{equation}
for weights $\pi_i\in[0,1]$, $i \in \M$, that satisfy $\sum_{i \in \M} \pi_i = 1$.

To derive the constrained barycentre measure and its characterisation, we first solve at the pure barycentre, that is Optimisation Problem \eqref{eq: barycentre-problem} without any constraints. This result might be of independent interest and will play a key role for the proposed constrained barycentre pricing measure.
\begin{theorem}[Barycentre]
\label{thm:barycentre}
Let $\Ex\left[\prod_{i \in \M} \left(\frac{d\Q_i}{d\P}\right)^{\pi_i}\right] < + \infty$. Then the barycentre measure, denoted by $\Qb$ and defined by
\begin{equation}
    \Qb  := \argmin_{\Q \ll \P} \sum_{i \in \M} \pi_i\,  D_{KL}(\Q ~||~\Q_i)
    \label{eq: pure-barycentre}
\end{equation}
exists and is unique. Moreover, its RN derivative to the reference measure $\P$ is given by 
\begin{equation*}
        \frac{d\Q^\mathfrak{b}}{d\P} = \frac{\prod_{i \in \M} \left(\frac{d\Q_i}{d\P}\right)^{\pi_i}}{\Ex\left[\prod_{i \in \M} \left(\frac{d\Q_i}{d\P}\right)^{\pi_i}\right]}\,.
    \end{equation*}    
\end{theorem}
\textit{Proof.}
Denote by $\Q^b$ the probability measure induced by the RN density $\frac{d\Q^b}{d\P} = \frac{\prod_{i \in \M} \left(\frac{d\Q_i}{d\P}\right)^{\pi_i}}{\Ex\left[\prod_{i \in \M} \left(\frac{d\Q_i}{d\P}\right)^{\pi_i}\right]}$, and note that $\Q^b$ is, by assumption, well-defined.

Next, we show that the following equivalency holds
\begin{equation}\label{eq:WKL-KL-c}
    \sum_{i \in \M} \pi_i D_{KL}(\Q ~||~\Q_i) = D_{KL}(\Q ~||~ \Q^b) + c \;,
\end{equation}
where $c := - \log\left(\Ex\left[\prod_{i \in \M}\left(\frac{d\Q_i}{d\P}\right)^{\pi_i}\right]\right)$.  We calculate
\begin{align*}
    \sum_{i \in \M} \pi_i D_{KL}(\Q ~||~\Q_i) &= \sum_{i \in \M} \pi_i\; \Ex^\Q\left[\log\left(\frac{d\Q}{d\Q_i}\right)\right]\\
        &= \sum_{i \in \M} \pi_i \; \Ex^\Q\left[\log\left(\frac{d\Q}{d\P} \frac{d\P}{d\Q_i}\right)\right]\\
        &= \Ex^\Q\left[\log\left(\frac{d\Q}{d\P}\right)\right] - \sum_{i \in \M} \pi_i \; \Ex^\Q\left[ \log\left(\frac{d\Q_i}{d\P}\right)\right]\\
        &= \Ex^\Q\left[\log\left(\frac{d\Q}{d\P}\right)\right] - \Ex^\Q\left[ \log\left( \prod_{i \in \M} \left(\frac{d\Q_i}{d\P}\right)^{\pi_i}\right)\right]\\
        &= \Ex^\Q\left[\log\left(\frac{d\Q}{d\P}\right)\right] - \Ex^\Q\left[ \log\left( \frac{\prod_{i \in \M} \left(\frac{d\Q_i}{d\P}\right)^{\pi_i}}{\Ex\left[\prod_{i \in \M} \left(\frac{d\Q_i}{d\P}\right)^{\pi_i}\right]} \Ex\left[\prod_{i \in \M} \left(\frac{d\Q_i}{d\P}\right)^{\pi_i}\right]\right)\right]\\
        &= \Ex^\Q\left[\log\left(\frac{d\Q}{d\P}\right)\right] - \Ex^\Q\left[ \log\left( \frac{\prod_{i \in \M} \left(\frac{d\Q_i}{d\P}\right)^{\pi_i}}{\Ex\left[\prod_{i \in \M} \left(\frac{d\Q_i}{d\P}\right)^{\pi_i}\right]}\right)\right] + c\\
        &= \Ex^\Q\left[\log\left(\frac{d\Q}{d\P}\right)\right] - \Ex^\Q\left[ \log\left(\frac{d\Q^b}{d\P}\right)\right] + c\\
        &= \Ex^\Q\left[ \log\left(\frac{d\Q}{d\P} \frac{d\P}{d\Q^b}\right)\right] + c\\
        &= D_{KL}(\Q ~||~ \Q^b) + c\;.
\end{align*}
As the geometric average is concave, we apply Jensen's inequality to obtain
\begin{equation*}
    c= 
    - \log\left(\Ex\left[\prod_{i \in \M}\left(\frac{d\Q_i}{d\P}\right)^{\pi_i}\right]\right)
    >
    -
    \log\left(\prod_{i \in \M}\Ex\left[\frac{d\Q_i}{d\P}\right]^{\pi_i}\right)
    = 
    \log(1) = 0\;
\end{equation*}
and thus \eqref{eq:WKL-KL-c} holds.

Therefore, \Cref{eq: pure-barycentre} becomes
\begin{equation*}
    \argmin_{\Q \ll \P} \sum_{i \in \M} \pi_i D_{KL}(\Q ~||~\Q_i) = \argmin_{\Q \ll \P}\;  D_{KL}(\Q ~||~ \Q^b) + c\;,
\end{equation*}
which is uniquely attained at $\Q^b$ as $c>0$ and the KL-divergence is strictly convex in the first argument. Thus, indeed $\Q^b$ is the barycentre and we write $\Qb = \Q^b$.
\hfill \Halmos

We observe that $\frac{d\Q^\mathfrak{b}}{d\P}$ is a geometric average of the RN derivatives $\frac{d \Q_i}{\P}$. That the KL barycentre is a geometric average has been obtained in e.g., \cite{NielsenNock2009} for densities.

Next, we derive the constrained barycentre pricing measure. 

\begin{theorem}[Constrained barycentre]
\label{thm:constrained-barycentre}
Let $\Q_i$, $i \in \M$ exist and denote by $\Qd$ the solution to Optimisation Problem \eqref{eq: barycentre-problem}. Then, $\Qd$ exists and is unique, and its RN derivative has representation
\begin{equation*}
\label{eq: RN-const-bary}
    \frac{d\Qd}{d\P} = \frac{\e^{-\lambda_\x \,Y} \;\prod_{i \in \M}\left(\frac{d\Q_i}{d\P}\right)^{\pi_i}}{\Ex\left[\e^{-\lambda_\x \,Y} \;\prod_{i \in \M}\left(\frac{d\Q_i}{d\P}\right)^{\pi_i}\right]}\;, 
\end{equation*}
where $\lambda_\x \in \R$ is the optimal Lagrange multiplier satisfying $\Ex^{\Qd}[Y] = \Ex^\P[Y]$ for all $\x \in \supp{\X}$, and whenever $\Ex\left[\e^{-\lambda_\x Y } \prod_{i \in \M}\left(\frac{d\Q_i}{d\P}\right)^{\pi_i}\right] < +\infty$.
\end{theorem}
\textit{Proof.} 
Denote by $\Q'$ the probability measure induced by the RN density 
$$\frac{d\Q'}{d \P}:= \frac{\e^{-\lambda_\x \,Y} \;\prod_{i \in \M}\left(\frac{d\Q_i}{d\P}\right)^{\pi_i}}{\Ex\left[\e^{-\lambda_\x \,Y} \;\prod_{i \in \M}\left(\frac{d\Q_i}{d\P}\right)^{\pi_i}\right]},$$ and note that $\Q'$ is by assumption well-defined.

Next, let $\Qtilde$ be an arbitrary probability measure such that $\Qtilde \ll \P$, $\Ex^{\Qtilde}[Y] = \Ex^\P[Y]$, and $\Qtilde(B) \neq \Q'(B)$ for at least one $B \in \F$.  We will show the following inequality, 
\begin{equation}
\label{eq: barycentre-inequality}
    \sum_{i\in \M} \pi_i D_{KL}(\Q' ~||~ \Q_i) <  \sum_{i\in \M} \pi_i D_{KL}(\Qtilde ~||~ \Q_i) \;,
\end{equation}
which implies that $\Q'$ is the solution to Optimisation Problem \eqref{eq: barycentre-problem}. 

To show \eqref{eq: barycentre-inequality}, we first, observe that 
\begin{equation}\label{eq:dq'-dp}
\frac{d\Q'}{d\P} = \frac{d\Qb}{d\P} \e^{-\lambda_\x Y} M\,,
\end{equation}
where $M := \frac{\Ex\left[\prod_{i = \M} \left(\frac{d\Q_i}{d\P}\right)^{\pi_i}\right]}{\Ex\left[\e^{-\lambda_\x Y}\prod_{i = \M} \left(\frac{d\Q_i}{d\P}\right)^{\pi_i}\right]}$ is a constant.
Moreover, it holds using \eqref{eq:dq'-dp} in the first equation, that
\begin{align*}
    \Ex^\Qb\left[\left(\frac{d\Q'}{d\Qb} - \frac{d\Qtilde}{d\Qb}\right) \log\left(\frac{d\Q'}{d\Qb}\right)\right] 
    &= \Ex^\Qb\left[\left(\frac{d\Q'}{d\Qb} - \frac{d\Qtilde}{d\Qb}\right) \log\left(\e^{-\lambda_\x Y } M\right)\right]\\
    &= \Ex^\Qb\left[\left(\frac{d\Q'}{d\Qb} - \frac{d\Qtilde}{d\Qb}\right) \big(-\lambda_\x Y + \log\left(M\right)\big)\right]\\
    &= -\lambda_\x \left(\Ex^{\Q'}[Y] - \Ex^{\Qtilde}[Y] \right) + \log(M) (1-1)\\
    &= 0\;,
\end{align*}
recalling that $\Q'$ and $\Qtilde$ satisfy the expectation constraint. Using \eqref{eq:WKL-KL-c} and that the above expectation is zero, the RHS of \Cref{eq: barycentre-inequality} is
\begin{align*}
    \sum_{i \in \M} \pi_i D_{KL}
(\Qtilde ~||~ \Q_i) &= D_{KL}(\Qtilde ~||~ \Qb) + c + \Ex^\Qb\left[\left(\frac{d\Q'}{d\Qb} - \frac{d\Qtilde}{d\Qb}\right) \log\left(\frac{d\Q'}{d\Qb}\right)\right] \\
&= \Ex^\Qb\left[\frac{d\Qtilde}{d\Qb} \log\left(\frac{d\Qtilde}{d\Qb}\right)\right] + \Ex^\Qb\left[\left(\frac{d\Q'}{d\Qb} - \frac{d\Qtilde}{d\Qb}\right) \log\left(\frac{d\Q'}{d\Qb}\right)\right] + c\\
&= \Ex^\Qb\left[\frac{d\Qtilde}{d\Qb}\left(\log\left(\frac{d\Qtilde}{d\Qb}\right) - \log\left(\frac{d\Q'}{d\Qb}\right)\right) + \frac{d\Q'}{d\Qb}\log\left(\frac{d\Q'}{d\Qb}\right)\right] + c \\
&= \Ex^\Qb\left[\frac{d\Qtilde}{d\Qb}\left(\log\left(\frac{d\Qtilde}{d\Qb} \frac{d\Qb}{d\Q'}\right)\right) + \frac{d\Q'}{d\Qb}\log\left(\frac{d\Q'}{d\Qb}\right)\right] + c \\
&= \Ex^\Qtilde\left[\log\left(\frac{d\Qtilde}{d\Q'}\right)\right] + \Ex^{\Q'}\left[\log\left(\frac{d\Q'}{d\Qb}\right)\right] + c\\
&= D_{KL}(\Qtilde ~||~ \Q') + D_{KL}(\Q' ~||~ \Qb) + c
\;. \label{eq: constrained-barycentre}
\end{align*}
Finally, again using \eqref{eq:WKL-KL-c}, and since $\Qtilde$ is not equivalent to $\Q'$,
\begin{align*}
D_{KL}(\Qtilde ~||~ \Q') + c + D_{KL}(\Q' ~||~ \Qb)
    &= D_{KL}(\Qtilde ~||~ \Q') + c +  \sum_{i \in \M} \pi_i \; D_{KL}(\Q' ~||~ \Q_i) - c\\
    &= D_{KL}(\Qtilde ~||~ \Q') +  \sum_{i \in \M} \pi_i \; D_{KL}(\Q' ~||~ \Q_i)\\
    &> \sum_{i \in \M} \pi_i \; D_{KL}(\Q' ~||~ \Q_i)\;,
\end{align*}
as claimed, and \eqref{eq: barycentre-inequality} holds. Thus, indeed $\Q'$ is the unique solution to \Cref{eq: barycentre-inequality}, and we use the notation $\Q' =\Qd$.

We further assert that the Lagrange multiplier $\lambda_x \in \R$ exists and is unique, as a consequence of \Cref{thm:existence}.
\hfill \Halmos

\begin{corollary}
    Denote $K_Y^\Qb(t):= \log\left(\Ex^\Qb\left[\e^{t(Y - \Ex[Y])}\right]\right)$, the cgf of $Y - \Ex[Y]$ conditional on $\X = \x$ under the barycentre measure $\Qb$. Then, 
    \begin{enumerate}[label = $\roman*)$]
        \item there exists a unique solution $\lambda_x^* \in \R$ to
    \begin{equation*}
        \partial_{t} \; K_Y^\Qb(t) \big\vert_{t = -\lambda_x} = 0\;,
    \end{equation*}
    \item and moreover $\lambda_x^\star$ is the optimal Lagrange multipliers of the RN derivative in  \Cref{thm:constrained-barycentre}.
    \end{enumerate}
\end{corollary}

\textit{Proof.}
Part i) follows immediately from \Cref{thm:existence}.

For part ii), first note that
\begin{align*}
    \frac{d\Qd}{d\Qb} 
    &= \e^{-\lambda_\x Y} \frac{\Ex\left[\prod_{i = \M} \left(\frac{d\Q_i}{d\P}\right)^{\pi_i}\right]}{\Ex\left[\e^{-\lambda_\x Y}\prod_{i = \M} \left(\frac{d\Q_i}{d\P}\right)^{\pi_i}\right]}
    \\
    &= \e^{-\lambda_\x Y} \left(\Ex\left[\e^{-\lambda_\x Y}\frac{\prod_{i = \M} \left(\frac{d\Q_i}{d\P}\right)^{\pi_i}}{\Ex\left[\prod_{i = \M} \left(\frac{d\Q_i}{d\P}\right)^{\pi_i}\right]}\right]\right)^{-1}
        \\
    &= \e^{-\lambda_\x Y} \left(\Ex\left[\e^{-\lambda_\x Y}\frac{d\Qb}{d\P}\right]\right)^{-1}
    \\
    &=
    \frac{\e^{-\lambda_\x Y}}{\Ex^{\Qb}[\e^{-\lambda_\x Y}]}\;.
\end{align*}
Assume, without loss of generality, that $\Ex[Y] = 0$. Then, the expectation constraint is
\begin{equation*}
    \Ex^\Qd[Y] = \Ex^\Qb\left[\frac{d\Qd}{d\Qb} Y\right] =\Ex^\Qb\left[\frac{Y \e^{-\lambda_\x Y}}{\Ex^{\Qb}[\e^{-\lambda_\x Y}]}\right] =  \frac{\E^\Qb\left[Y \e^{-\lambda_x Y}\right]}{\E^\Qb\left[\e^{-\lambda_x Y}\right]} = \partial_{t} \; K_Y^\Qb(t)\big|_{t = -\lambda_x} = 0\;,
\end{equation*}
as claimed.
\hfill \Halmos

\section{Application of Discrimination-Insensitive Pricing}
\label{sec: application}
Here, we consider a case study in which we compare the constrained barycentre model to the marginally discrimination-insensitive models, and to the discrimination-insensitive model as described in \Cref{sec: disc-ins-pricing-measure}.
\subsection{Distribution of covariates}
 To do this, we generate a dataset to represent an automotive insurance portfolio. The policyholders' covariates are $(D_1,D_2, X_1, X_2)$, where the first two covariates are the  discriminatory ones and the latter two are permitted covariates. We describe their joint distribution under the real-world probability measure $\P$. The protected  covariate $D_1$ models the attribute gender and its distribution satisfies $\P(D_1 = 1) = 1 - \P(D_1 = -1) = 0.4$, where we identify  $-1$ with ``women" and $1$ with ``men". Thus, 40\% of the dataset corresponds to ``men''. $D_2$ represents ``marital" status, and takes values $-1$ with probability $\frac{3}{8}$, $0$ with probability $\frac{1}{8}$, and $1$ with probability $\frac{1}{2}$, representing three categories ``married", ``single" and ``other", respectively.  The permitted covariate $X_1$ represents the categories of hours driven. Specifically $X_1$ takes values $-1$, $0$ and $1$ that represent ``few", ``some", and ``many" hours driven, respectively.  We take $X_1$ to depend on $D_1$ in the following manner: 
\begin{align*}
    X_{1}|_{D_1 = -1} \;\; &= \begin{cases}
    -1 & \text{ with prob.} \quad  \frac{1}{2}\\
    0 & \text{ with prob.} \quad\frac{1}{4}\\
    1 & \text{ with prob.} \quad\frac{1}{4}
    \qquad \qquad \text{and}
\end{cases}
\\[1em]
X_1|_{D_1 = 1} \;\;&= \begin{cases}
    -1 & \text{ with prob.}\quad \frac{1}{4}\\
    0 & \text{ with prob.} \quad\frac{1}{4}\\
    1 & \text{ with prob.} \quad \frac{1}{2}\,.
\end{cases}
\end{align*}
This, for example, means that in this dataset women are more likely in the category of driving ``few'' hours and less likely in the category of driving ``many'' hours. 

Representing the value of the vehicle, in 1000s, the covariate $X_2$ is dependent on $D_2$ through
\begin{align*}
 &X_2|_{D_2 = -1} \sim \text{tExp}\big(\text{rate} = \tfrac{1}{35},\; \text{ub} = 120\big)\\
 &X_2|_{D_2 = 0} \phantom{-}\sim \text{tExp}\big(\text{rate} = \tfrac{1}{50},\; \text{ub} = 90\big)
 \\
 &X_2|_{D_2 = 1} \phantom{-} \sim\;  \text{tNorm}\big(\text{m} = 50, \text{sd} = 15,\; \text{lb} = 1.5,\; \text{ub} = 100\big)
\end{align*}

where tExp and tNorm are the truncated exponential and normal distributions with respective parameters, and the lower/upper bounds (lb respectively ub) denote the truncation levels.  The error term $\mfu$ is a standard normal rv, and the simulated losses are given by $Y: = h(D_1, D_2, X_1, X_2) + \mfu$, where $h(d_1, d_2, x_1, x_2) =  15 + 3 x_1 + \frac{1}{4} x_2 + 5 d_1 + 3 d_2$. A dataset of 500,000 samples is generated, then a root finding procedure is carried out to find the corresponding optimal Lagrange multipliers. 50 samples are accepted based on convergence criteria of the roots, and average values of the premia and sensitivities are reported, in order to mitigate numerical error from the optimisation.

\subsection{Comparison of pricing measures}

For calculating the insensitive prices, we consider the ES risk-loaded expected value premium, as detailed in \Cref{ex: ES-risk-load} with ES tolerance level $\alpha = 0.9$ and risk-loading parameter $c = 0.2$. We denote by $\Q_1$ (respectively $\Q_2$) the marginal discrimination-insensitive measure with only the sensitivity constraint for $D_1$ (respectively $D_2$), that is the solution to \Cref{eq: i-problem}. $\Q^\dagger$ is the constrained barycentre pricing measure between $\Q_1$ and $\Q_2$, with weights equal to $\pi_1 = \pi_2 = \frac{1}{2}$, and expectation constraint, see \Cref{thm:constrained-barycentre}. Finally, $\Q^*$ is the discrimination-insensitive pricing measure with sensitivity constraint $D_1$ and $D_2$, and $\Q^*_1$ is the discrimination-insensitive pricing measure with sensitivity constraint $D_1$ (respectively $\Q^*_2$ with $D_2$), see \Cref{sec: disc-ins-pricing-measure}. To quantify the discrimination present in a model, we define the \textit{marginal sensitivity} and \textit{total sensitivity} as follows.

\begin{figure}[th]
    \centering
    \includegraphics[width = 0.6\textwidth]{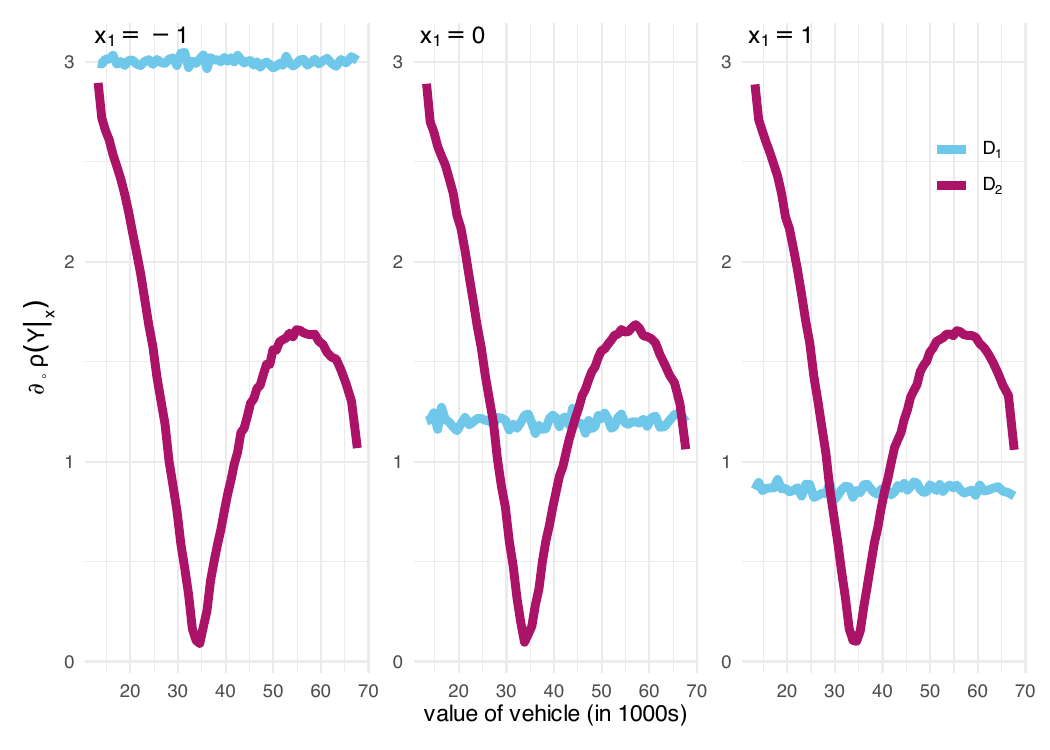}
    \caption{Sensitivities $\partial_{D_i}\rho^{\P}(Y|_x)$, $i = 1,2$ conditional jointly on $X_1$ (hours driven) and $X_2$ (vehicle value) under the reference measure $\P$. Left panel $X_1 =- 1$ (few); centre panel $X_1 = 0$ (some); right panel $X_1 = 1$ (many), and $x$-axis corresponds to realisations of $X_2$. Blue curve corresponds to the sensitivity to $D_1$ and the purple to the sensitivity to $D_2$.}
    \label{fig: sens_P}
\end{figure}

\begin{definition}[Marginal and total sensitivity]
     For a probability measure $\Q\ll \P$, we define the (absolute value of) marginal sensitivity of a distortion premium principle with respect to the $i$th protected covariate as
     \begin{equation*}
         \xi_i(\Q) := \int_{\supp{\X}} \abs{\partial_{D_i} \rho^\Q(Y|_\x)}\; dF_\X(\x)\;.
     \end{equation*}

     Then, the total sensitivity of a distortion premium principle is 
\begin{equation*}
    \xi(\Q) := \sum_{i = 1}^m \xi_i(\Q)\;.
\end{equation*}
\end{definition}

 \Cref{fig: sens_P} displays the total sensitivity of the premia conditional on the permitted covariates $\rho(Y|_\x)$ under the reference measure $\P$, with the left, centre and right panels corresponding to the hours driven of $x_1 = -1$, $x_1 = 0$, and $x_1 = 1$ respectively.  In this figure, we can see that the sensitivity of the pricing principle under $\P$ is, in general, non-zero.  Here, we can see that the sensitivity stemming from $D_1$ (gender) is approximately constant over the value of the vehicle, for each category of hours driven, though with varying values found. The sensitivity with respect to $D_2$ (marital status) is non-constant over the value of the vehicle, and has a similar curve for all categories of hours driven. These sensitivities stem from the construction of the simulated losses, as $D_1$ has a linear relationship with the losses $Y$ through $X_1$ being categorically dependent on $D_1$, while $Y$ has complex dependency on $D_2$ via $X_2$.

\begin{figure}[h]
    \centering
    \includegraphics[width = 0.6\textwidth]{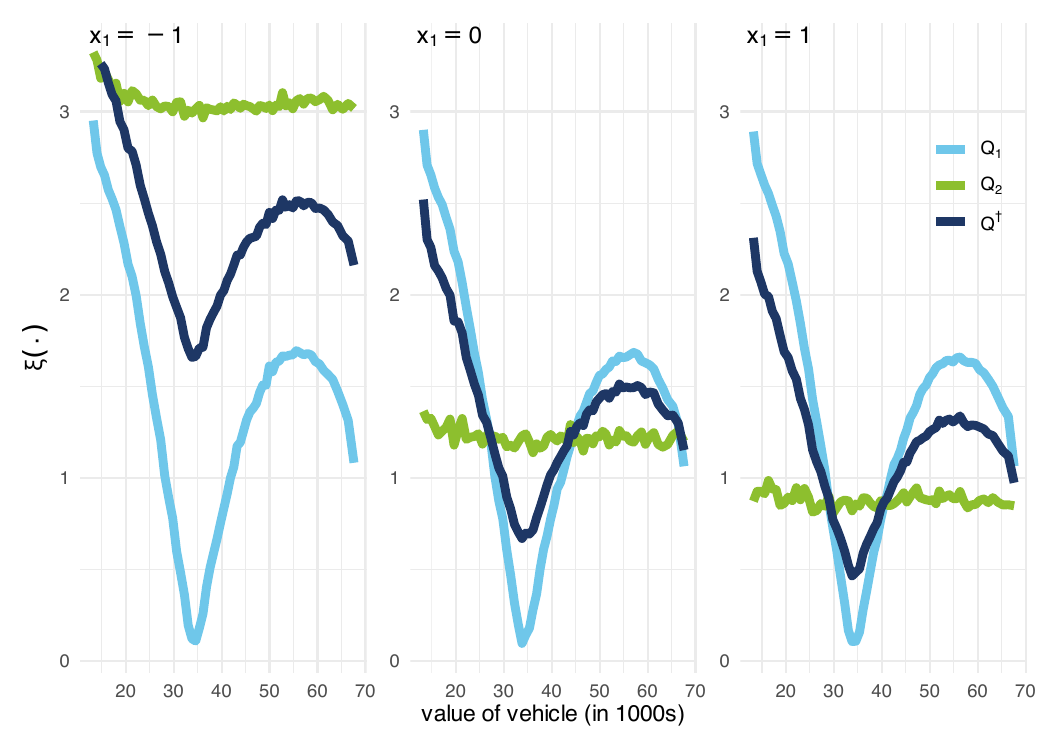}
    \caption{Total sensitivities $\sum_{i = 1}^2\partial_{D_i}\rho^{\circ}(Y|_x)$ conditional jointly on $X_1$ (hours driven) and $X_2$ (vehicle value) under the marginally-insensitive pricing measures, $\Q_1$, and $\Q_2$, and the constrained barycentre measure $\Qd$. Light blue curve corresponds to the sensitivity under measure $\Q_1$, green to the sensitivity under $\Q_2$, and dark blue to the sensitivity under $\Q^\dagger$.}
    \label{fig: sens_Qd}
\end{figure}

Figures \ref{fig: sens_Qd} and \ref{fig: sens_Qstar} displays the total sensitivity of the premia conditional on both permitted covariates in the same panel convention as \Cref{fig: sens_P}.  By construction of the respective measures, $\Q_1$ and $\Q^*_1$ are insensitive to $D_1$, likewise $\Q_2$ and $\Q^*_2$ are insensitive to $D_2$.  Hence, the total sensitivity only has contributions from the other protected covariate.  In \Cref{fig: sens_Qd}, the marginally-insensitive and constrained barycentre pricing measures all retain some amount of sensitivity, as these measure are not constrained to have total sensitivity $0$. This is similarly true in \Cref{fig: sens_Qstar}, where $\Q^*_1$ and $\Q^*_2$ have non-zero sensitivity. This is in contrast to $\Q^*$, which enforces zero sensitivity to both protected covariates, and thus has zero total sensitivity.

\begin{figure}[h]
    \centering
    \includegraphics[width = 0.6\textwidth]{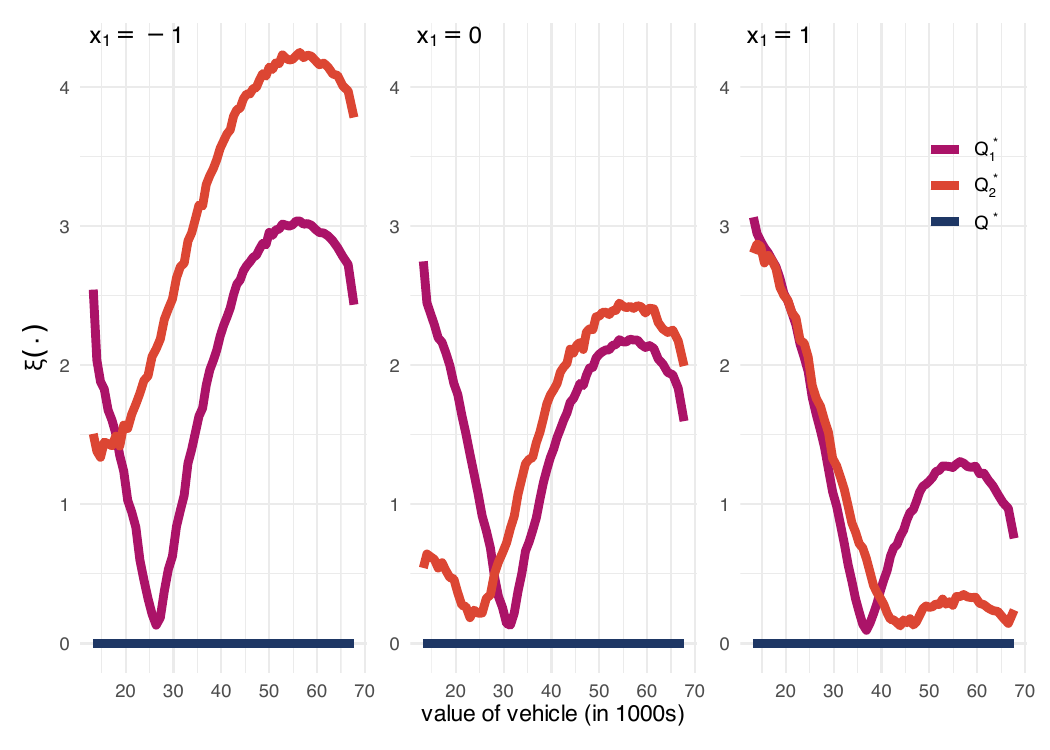}
    \caption{Total sensitivities $\sum_{i = 1}^2\partial_{D_i}\rho^{\circ}(Y|_x)$ conditional jointly on $X_1$ (hours driven) and $X_2$ (vehicle value) under the discrimination-insensitive pricing measures, $\Q^*_1$, $\Q^*_2$, and $\Q^*$. Fuchsia curve corresponds to the sensitivity under measure $\Q^*_1$, red to the sensitivity under $\Q^*_2$, and dark blue to the sensitivity under $\Q^*$.}
    \label{fig: sens_Qstar}
\end{figure}

\begin{figure}[h]
    \centering
    \includegraphics[width = 0.6\textwidth]{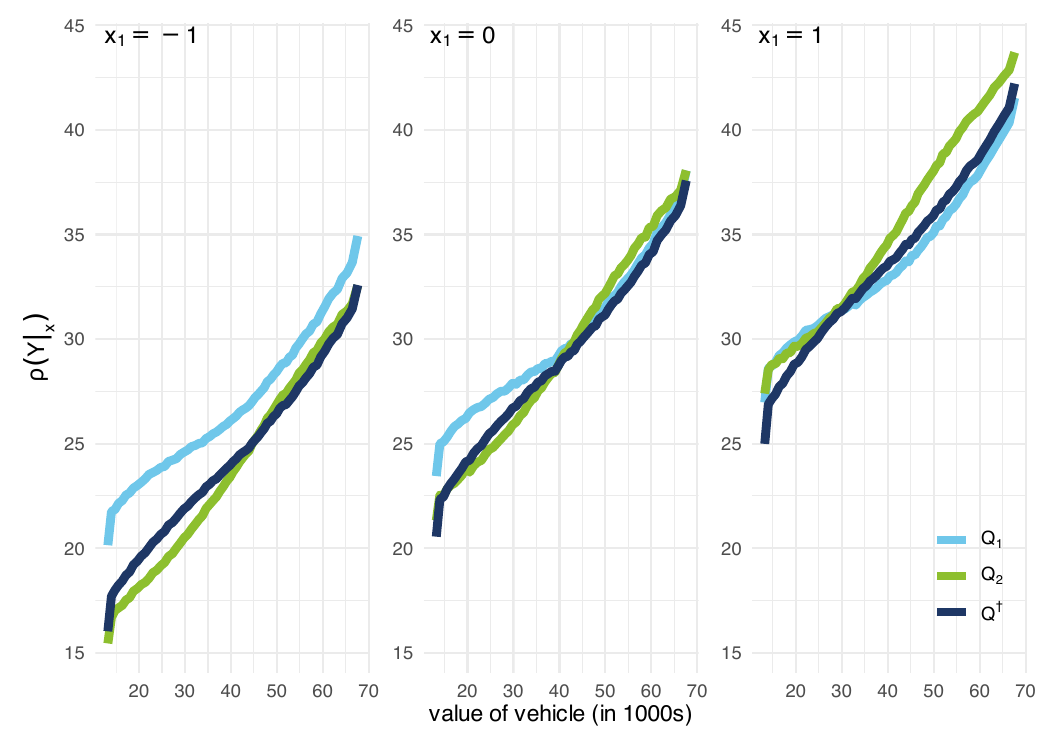}
    \caption{Premia $\rho^{\circ}(Y|_x)$ conditional jointly on $X_1$ (hours driven) and $X_2$ (vehicle value) under the marginally-insensitive pricing measures, $\Q_1$, and $\Q_2$, and the constrained barycentre measure $\Qd$. Light blue curve corresponds to the premium under measure $\Q_1$, green to the premium under $\Q_2$, and dark blue to the premium under $\Q^\dagger$.}
    \label{fig: rho_Qmarginals_Qdag}
\end{figure}

Figures \ref{fig: rho_Qmarginals_Qdag} and \ref{fig: rho_Qstars} display the premia under the measures discussed, with the same panel convention as previous figures.  In all three panels, the premium is monotonically increasing with respect to the permitted covariate $X_2$, consistent with the loss being increasing in $X_2$.  In \Cref{fig: rho_Qmarginals_Qdag}, we can see that the premium under $\Qd$ is not a pure interpolation of $\Q_1$ and $\Q_2$, as we also enforce the conditional expectation constraint, which is absent from $\Q_1$ and $\Q_2$. The premium under the constrained barycentre pricing measure appears flattened compared to the marginally-insensitive measures. This is in contrast to \Cref{fig: rho_Qstars}, in which the conditional premia appear less linear than in \Cref{fig: rho_Qmarginals_Qdag}.

\begin{figure}[h]
    \centering
    \includegraphics[width = 0.6\textwidth]{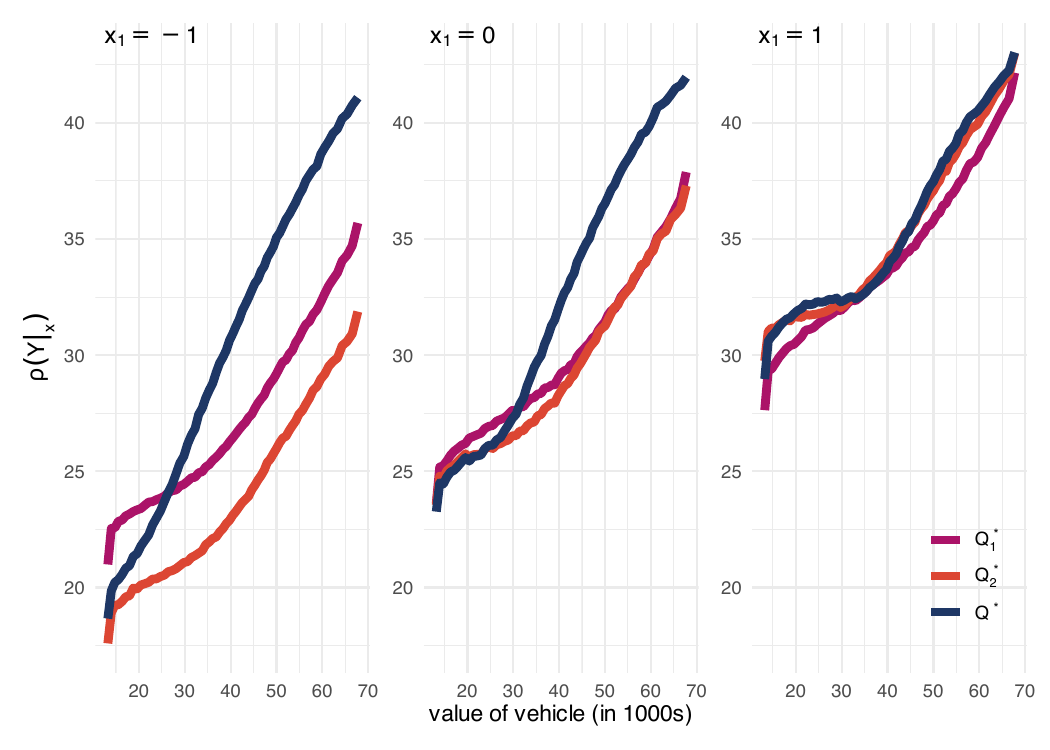}
    \caption{Premia $\rho^{\circ}(Y|_x)$ conditional jointly on $X_1$ (hours driven) and $X_2$ (vehicle value) under the discrimination-insensitive pricing measures, $\Q^*_1$, $\Q^*_2$, and $\Q^*$. Fuchsia curve corresponds to the premium under measure $\Q^*_1$, red to the premium under $\Q^*_2$, and dark blue to the premium under $\Q^*$.}
    \label{fig: rho_Qstars}
\end{figure}

\begin{table}[h]
\centering
\begin{tabular}{@{}lccccccc@{}}

& $\P$  & $\Q^*_1$ & $\Q^*_2$ & $\Q^*$ & $\Q_1$ & $\Q_2$ & $\Qd(0.5, 0.5)$ \\[0.5em] 
\toprule\toprule
$\xi_1(\cdot)$ &  1.69  & 0    &  1.89 & 0  & 0   &   1.72    &   0.89\\[0.5em]
$\xi_2(\cdot)$& 1.39  & 1.65    &0 &   0    &    1.40    &   0   &   0.76\\[0.5em]
$\xi(\cdot)$ & 3.07   & 1.65    & 1.89 &  0    & 1.40    &  1.72    &   1.64\\ \bottomrule
\end{tabular}
\caption{Marginal and total sensitivities to the discriminatory covariates under varying measures.}
\label{tab: sensitivities}
\end{table}

\begin{figure}[h]
    \centering
    \includegraphics[width = 0.6\textwidth]{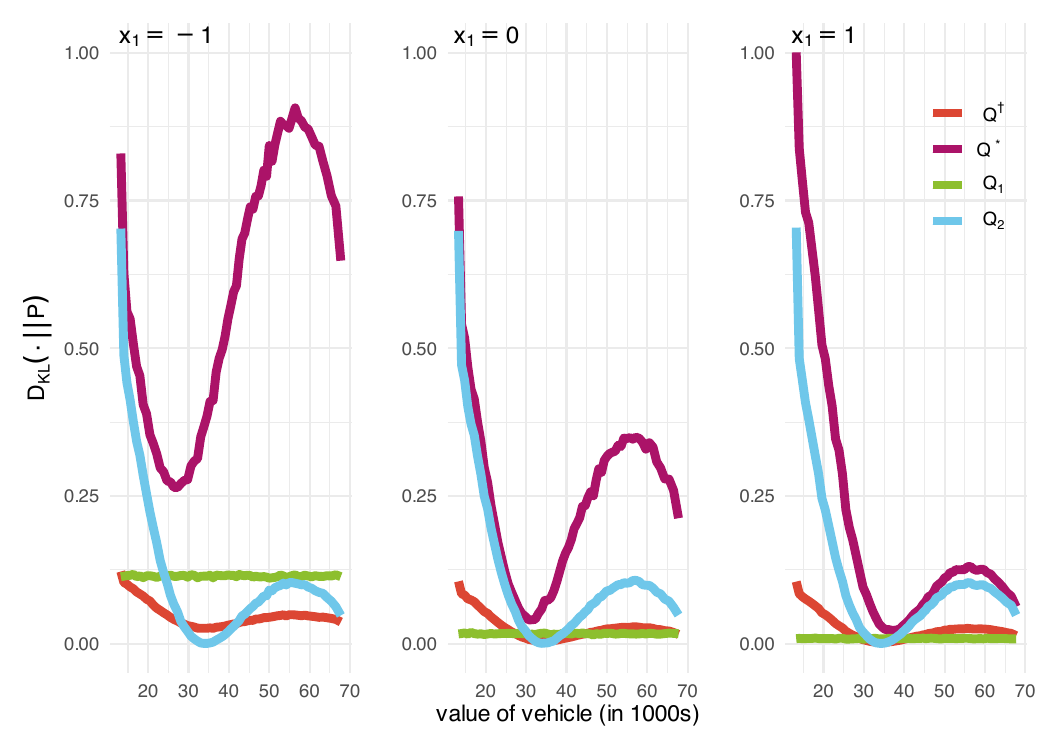}
    \caption{KL-divergence between varying measures and the reference measure $\P$, conditional jointly on $X_1$ (hours driven) and $X_2$ (vehicle value) for $\Q_1$ (green), $\Q_2$ (blue), $\Q^*$ (fuchsia), and $\Q^\dagger$ (red).}
    \label{fig: KL}
\end{figure}

\Cref{tab: sensitivities} reports the sensitivities of the premium under the measures discussed.  Here, the reference measure $\P$ has the greatest overall sensitivity, while the discrimination-insensitive measure $\Q^*$ has, by construction, the least. Second smallest total sensitivity is of $\Q_1$, though the marginally-insensitive measure does not impose the expectation constraint. 
Additionally, \Cref{fig: KL} illustrates the KL divergence or ``cost" of the change-of-measure.  In particular, it can be seen here that enforcing both discrimination-insensitive constraints, as well as the expectation constraint in $\Q^*$ results in the largest divergence from the reference measure $\P$, when compared to the measures $\Q_1$, $\Q_1$, which enforce one discrimination-insensitivity constraint each, and $\Qd$, which only enforces an expectation constraint.
\vspace{1em}

\subsection{Choosing the barycentre weights}
\label{subsec: choosing-weights}
The choice of weights for use in the constrained barycentre pricing measure can be motivated by business goals or best practices.  However, if the weights are not a priori selected, we propose a proportional reduction of sensitivity as follows. We find the weights $(\pi_1, \pi_2)$ that reduce the sensitivity with respect to $D_1$ and $D_2$ proportionally to the overall sensitivity under the reference model $\P$. The weights can then be found by minimising the squared difference, as follows

\begin{equation*}
\min_{\substack{(\pi_1, \pi_2)\\[0.2em] \pi_1 + \pi_2 = 1}}\left(\frac{\abs{\xi_1 (\P) - \xi_1 (\Qd(\pi_1, \pi_2))}}{\xi_1(\P)} - \frac{\abs{\xi_2 (\P) - \xi_2 (\Qd(\pi_1, \pi_2))}}{\xi_2(\P)}\right)^2
\end{equation*}

Under this scheme in the scenario discussed in \Cref{sec: application}, the optimal weights are found to be $(0.58, 0.42)$.  We compare the sensitivities with choices of weights in \Cref{tab: sensitivities_weights}.
\begin{table}[h]
\centering
\begin{tabular}{@{}lccc@{}}

&$\Qd(0.2, 0.8)$ & $\Qd(0.58, 0.42)$ & $\Qd(0.8, 0.2)$ \\[0.5em] 
\toprule\toprule
$\xi_1(\cdot)$ & 1.37 & 0.76  & 0.38  \\[0.5em]
$\xi_2(\cdot)$ &0.33 &  0.86 &  1.14 \\[0.5em]
$\xi(\cdot)$ &1.70 &   1.62  &  1.52 \\ \bottomrule
\end{tabular}
\caption{Marginal and total sensitivities to the discriminatory covariates under the constrained barycentre pricing measure, with varying weights.}
\label{tab: sensitivities_weights}
\end{table}
\begin{figure}[h]
    \centering
    \includegraphics[width = 0.6\textwidth]{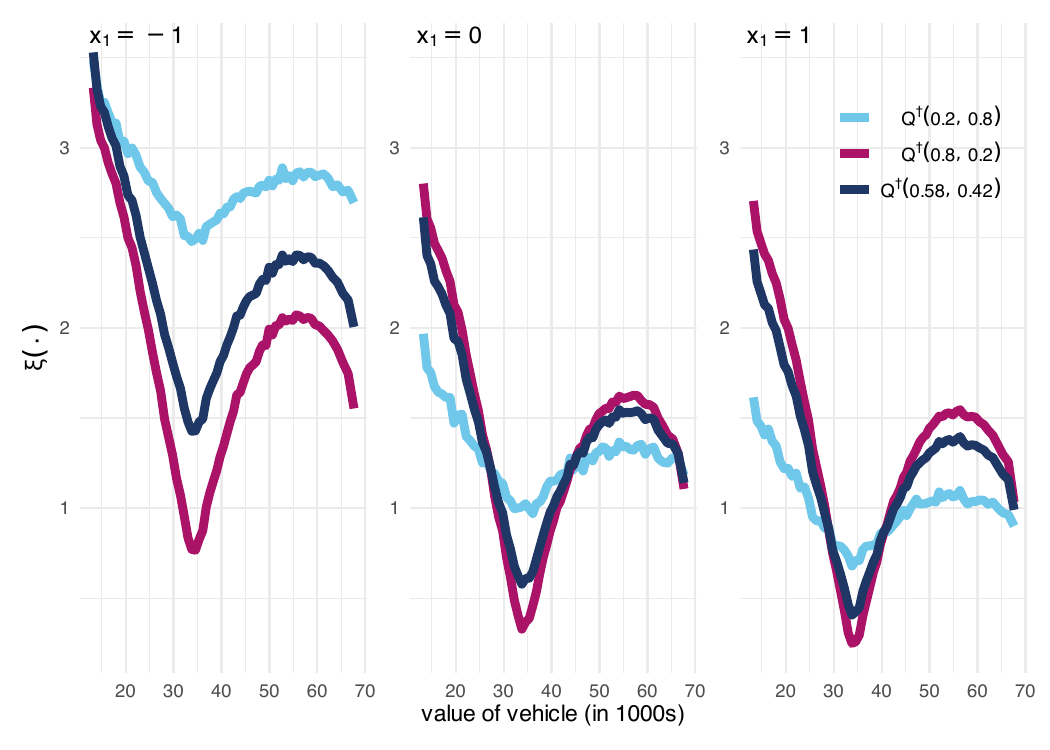}
    \caption{Total sensitivities $\sum_{i = 1}^2\partial_{D_i}\rho^{\circ}(Y|_x)$ conditional jointly on $X_1$ (hours driven) and $X_2$ (vehicle value) under the discrimination-insensitive pricing measures, $\Q^*_1$, $\Q^*_2$, and $\Q^*$. Fuchsia curve corresponds to the sensitivity under measure $\Q^*_1$, red to the sensitivity under $\Q^*_2$, and dark blue to the sensitivity under $\Q^*$.}
    \label{fig: sens_Qd_weights}
\end{figure}

\begin{figure}[th]
    \centering
    \includegraphics[width = 0.6\textwidth]{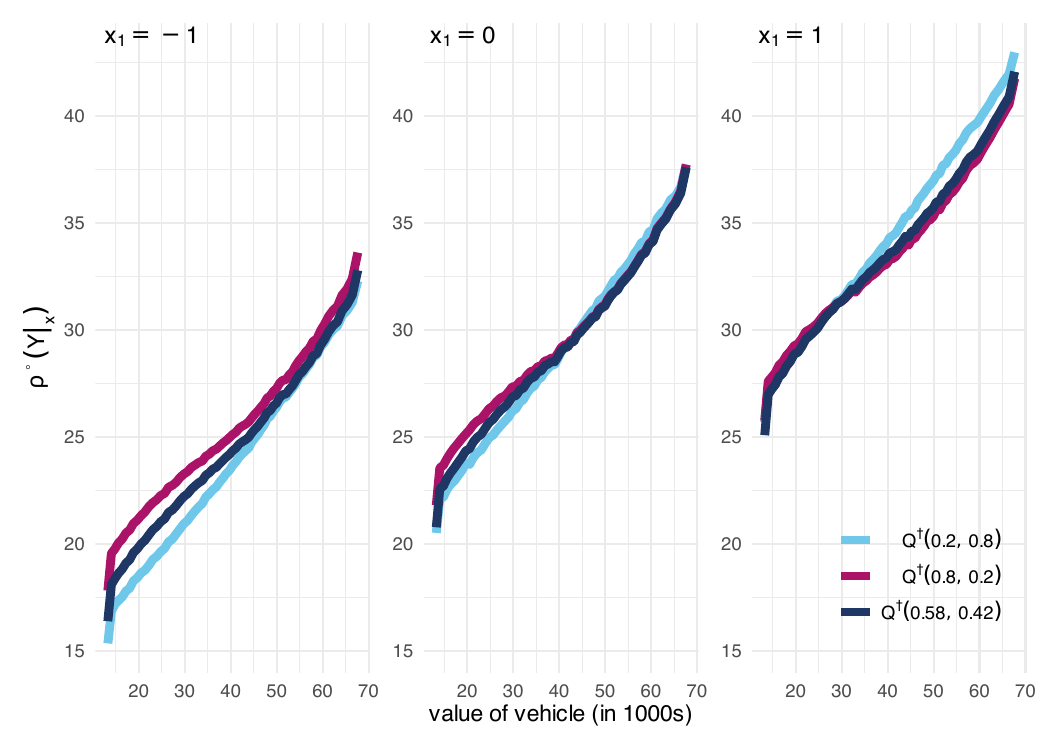}
    \caption{Premia $\rho^{\circ}(Y|_x)$ conditional jointly on $X_1$ (hours driven) and $X_2$ (vehicle value) under the constrained barycentre pricing measure $\Qd$ with varying weights. Light blue curve corresponds to the premium under measure $\Qd(0.2, 0.8)$, fuchsia to the premium under $\Qd(0.8,0.2)$, and dark blue to the premium under $\Qd(0.58, 0.42)$.}
    \label{fig: rho_Qd_weights}
\end{figure}

We observe that the total sensitivities of the constrained barycentre pricing measures follow the weights chosen, where if less weight is placed on $\Q_1$, then the resultant constrained barycentre pricing measure will have more total sensitivity. This is due to the sensitivity found in $\Q_1$ and $\Q_2$, where $\Q_2$ has less sensitivity than $\Q_1$.  Notably, this is not a pure interpolation between the two total sensitivities, as $\Qd$ must have the same conditional expectations as $\P$.

In Figures \ref{fig: sens_Qd_weights} and \ref{fig: rho_Qd_weights}, total sensitivities and premia under the barycentre pricing measure are plotted with varying weights.  In both plots, the sensitivity and premia preserve the same curve for all choices of weights.  A more extreme division of weights, i.e. one weight close to $1$, lead to more extreme sensitivities and premia.

\section{Conclusion}
This paper develops a methodology for pricing financial products that are insensitive to protected characteristics of the clients, such as age or ethnicity. By formulating an optimisation problem that finds the nearest (in Kullback-Leibler divergence) measure that has sensitivity zero to the protected covariates and maintains the same conditional expectation as the reference measure, we derive a novel pricing measure, termed the ``discrimination-insensitive measure". We show that this discrimination-insensitive measure exists and is unique under mild conditions.  In settings where the discrimination-insensitive measure is inappropriate for use in pricing, we propose a two-step procedure to that first finds ``marginally-insensitive" measures, which are then reconciled via a constrained barycentre model.  Here, representation, existence, and uniqueness are also established for the constrained barycentre measure.  As a intermediary step, we prove the solution of the pure (right-) Kullback-Leibler divergence barycentre in a general probability measure setting. Numerical experiments compare the discrimination-insensitive, marginally-insensitive, and constrained barycentre measures to one another, and to other fairness-seeking premia in the literature.

\ACKNOWLEDGMENT{The authors thank the Canadian Institute of Actuaries (CIA) for their financial support. KM further acknowledges support from the Ontario Graduate Scholarship and the Spencer Education Foundation. SP gratefully acknowledges support from the Natural Sciences and Engineering Research Council of Canada (grant RGPIN-2025-05847) and from the Data Science Institute at the University of Toronto (DSI-CGY3R1P03).}

\bibliographystyle{informs2014}
\bibliography{Refs.bib}
\end{document}